\documentclass[runningheads]{llncs}
\usepackage[T1]{fontenc}

\usepackage{newtxtext}
\usepackage[varvw]{newtxmath}

\usepackage{graphicx}
\usepackage{hyperref}
\usepackage{color}

\usepackage{comment}
\usepackage{amsmath}
\usepackage{mathtools}
\usepackage{bbm}
\usepackage{booktabs}
\usepackage{multirow}
\usepackage{tikz}
\usepackage{xcolor}
\usepackage{microtype}
\usepackage{float}
\usepackage{algorithm}
\usepackage[noend]{algpseudocode}
\usepackage[numbers,sort&compress]{natbib}
\usepackage{thmtools}
\usepackage{cleveref}
\usepackage{todonotes}
\makeatletter
\define@key{todonotes}{petr}[true]{\setkeys{todonotes}{color=orange!25}}
\define@key{todonotes}{martin}[true]{\setkeys{todonotes}{color=blue!25}}
\makeatother

\graphicspath{{figures/}}
\usepackage{wrapfig}

\usepackage{subcaption}
\usepackage{etoolbox}

\newtoggle{arxiv}
\toggletrue{arxiv}

\renewcommand{\cref}{\Cref}

\newif\ifshowcomments
\showcommentstrue

\newcommand{\E}{\mathbb{E}}
\newcommand{\Prob}{\mathbb{P}}
\newcommand{\R}{\mathbb{R}}
\newcommand{\N}{\mathbb{N}}
\newcommand{\cF}{\mathcal{F}}
\newcommand{\KL}{\mathrm{KL}}

\newcommand{\ind}[1]{\mathbbm{1}\!\left[#1\right]}
\DeclareMathOperator*{\argmax}{arg\,max}
\DeclareMathOperator*{\argmin}{arg\,min}
\DeclareMathOperator{\supp}{supp}
\DeclareMathOperator{\clip}{clip}

\newcommand{\cX}{\mathcal{X}}
\newcommand{\cG}{\mathcal{G}}
\newcommand{\cC}{\mathcal{C}}

\newcommand{\cP}{\mathcal{P}}
\newcommand{\cU}{\mathcal{U}}
\newcommand{\cA}{\mathcal{A}}
\newcommand{\qstar}{q^{*}}
\newcommand{\RIPr}{\mathrm{RIPr}}

\newcommand{\feat}{\varphi}

\newcommand{\monitor}{\nu}
\newcommand{\cH}{\mathcal{H}}
\newcommand{\HH}{H}

\newcommand{\dist}[1]{\Delta(#1)}

\newcommand{\model}{\delta}
\newcommand{\LM}{\mathrm{ML}}
\spnewtheorem*{excontinued}{Example}{\itshape}{\rmfamily}

\newenvironment{excont}[1]{\def\excontnum{\ref{#1}}\begin{excontinued}[continued]}{\end{excontinued}}

\begin{document}
\title{Anytime-Valid Monitoring of Robust Stochastic Models}
\nottoggle{arxiv} {
    \author{Hidden for the review process.}
} {
    \author{
        Martin Kurečka\inst{1} \and
        Roman Andriushchenko\inst{1,2} \and
        Milan Češka\inst{2} \and
        Nicola Paoletti\inst{3} \and
        Petr~Novotný\inst{1}
    }
}
\authorrunning{M. Kurečka et al.}
\institute{
    Masaryk University, Brno, Czechia \and
    Brno University of Technology, Brno, Czechia \and
    King's College London, London, UK
}
\maketitle
\begin{abstract}
Formal methods provide a principal approach to obtaining safety guarantees for autonomous agents. Their key bottleneck is their reliance on a formal model (e.g., a Markov decision process) describing the agent-environment interaction. Learning or crafting the model is an error-prone process, which undermines the guarantees offered by formal methods.
Robust models account for modelling uncertainty, but may still fail to capture the true environment.
This issue can be addressed through \emph{runtime model monitoring}, which aims to raise an alarm when a likely model-environment discrepancy is detected, while keeping the probability of false alarm bounded.
Existing approaches impose various assumptions
and do not clearly explain the rationale behind their design choices, leaving the broader design space for these monitors somewhat opaque.
In this paper, we formalise the design space of monitors for robust stochastic models through the lens of \emph{parameterised sequential betting} algorithms, highlighting the importance of \emph{observation feature} design (choosing what information about the history the monitor uses) and \emph{bet design} (how rich a family of betting rules the monitor can represent).
We show that our framework generalizes the existing approaches, and we use its formalization to propose a new \emph{mixture of linear monitors (ML)} design.
We prove that these ML monitors achieve an asymptotically optimal detection time within the class of linear monitors, making them a principled choice for our monitoring problem.
Our experiments confirm that ML monitors significantly reduce detection times compared to existing monitors and can effectively exploit more complex observation features.
\end{abstract}

\section{Introduction}
\label{sec:introduction}
\paragraph{Probabilistic Formal Methods: Promises and Bottlenecks}
Obtaining safety guarantees on the behavior of autonomous agents is a key objective of both formal methods and AI~\cite{seshia22CACMverifiedai,dalrymple2024guaranteedsafeaiframework}. Verification of sequential decision-making frameworks offers a well-developed toolbox of methods that can prove an agent's adherence to a specified objective~\cite{pmc,katoenMC}. A key bottleneck of this toolbox is its requirement for a \emph{model} of the agent--environment system, e.g., in the form of a Markov decision process (MDP).
The guarantees offered by such formal methods hold \emph{with respect to the model}, and any discrepancy between the model and the actual system might weaken or even invalidate the guarantees during real-world deployment. For complex domains and real-world interactions, model construction is inevitably error-prone: a model crafted by an expert, if feasible at all, is susceptible to oversights, while a model learned from real-world data is subject to statistical noise~\cite{robustanytime}. Consequently, agreement between the model and the modeled system is an important concern in contemporary formal-methods research~\cite{henzinger2025alignment}.

\paragraph{Anytime-Valid Model Monitoring}

A principled way to address model--system discrepancies is \emph{runtime model monitoring}~\cite{henzinger2025alignment}. A monitor observes the system's trajectory during deployment and raises an alarm when it detects a significant discrepancy between the model and the observed behavior. It thereby identifies situations in which model-based guarantees might no longer hold, prompting an appropriate response from, e.g., a supervisory safety-control layer or an overseeing authority.
A primary requirement of a monitor is that it is \emph{anytime \( \alpha \)-valid} for a given significance level \( \alpha \): this means that the probability of a false alarm (i.e., wrongly raising an alarm in the case that the true environment matches the model) along a trajectory sampled from the environment is at most \( \alpha \).
This guarantee concerns the whole, possibly unbounded, run.
This makes the monitoring problem an instance of the sequential hypothesis testing framework~\cite{ramdas2023gametheoretic}, with the model as the null hypothesis.

\paragraph{Monitoring Robust Models}
As noted above, a model learned from finite data is unlikely to capture the true environment exactly, due to statistical noise.
Thus, monitoring a single learned model is uninformative, since we are \emph{a priori} nearly certain that an alarm should be raised. However, learning procedures often produce a (potentially infinite) \emph{ensemble} of models that plausibly contains the true environment. Such ensembles can be represented finitarily, e.g., by a \emph{robust} or \emph{interval} MDP, which specifies, for each state-action pair, an \emph{uncertainty set} of plausible successor distributions~\cite{iyengar,Nilim2005RobustCO,wiessmanrmdp,GIVANinterval,intervaljair,rv26monitoring}.
The monitoring task is then to raise an alarm in cases when the true environment is not contained in the ensemble (a \emph{composite null} case in sequential hypothesis testing~\cite{sethi2026markovian,ramdas2023gametheoretic,Grunwald2024-pk,Lardy2024-gy}).

\paragraph{Our Motivation} Several approaches to runtime model monitoring have recently been developed \cite{sethi2026markovian,Fields2025-gu,henzinger2025alignment}
Given this literature, it is natural to ask why further model monitors are necessary. Our key motivation is to \emph{formalise the design landscape} of model monitors with a focus on effective monitoring of \emph{robust models}. The existing works impose diverse assumptions on the models or the environment and do not comprehensively expose the rationale behind the design choices of their monitors:

For instance, \cite{sethi2026markovian} considers monitoring of robust Markov chains against an ergodic Markovian environment. Their \emph{transition-wise} monitors estimate the environment transition probabilities in every state and use these estimates to detect model mismatch. While they prove the asymptotic optimality of their monitors w.r.t. the expected alarm raise time, their monitors require a large number of state revisits to produce estimates precise enough to provide evidence of model mismatch.
This can delay detection when discrepancies are systematic but small for individual transitions, or in episodic reinforcement learning tasks where each state may be visited only a few times, if at all.

On the other hand, the monitors in~\cite {henzinger2025alignment} observe the value of a fixed scoring rule that measures per-step discrepancies between model predictions and observed outcomes.
The scoring rule distills the discrepancies into a single scalar signal pooled across all states, which may enable faster gathering of evidence if the scoring rule consistently captures the mismatch in dynamics.
However, we show (\Cref{ex:tradeoff}) that a monitor based on scoring rules may be blind to certain misspecifications that are systematic, but are visible only through features that are not correlated with the transition probabilities predicted by the model.
As the authors state, their focus is on estimation,\footnote{As a matter of fact, the monitors of~\cite{henzinger2025alignment} do not provide a test for model-environment consistency, but a confidence sequence~\cite{howard2021timeuniform} for the average score.
}
and thus the paper does not provide a theoretical analysis of how scoring rule choice may affect detection power.
The resulting monitors come without explicit alarm time guarantees and do not support monitoring of robust models out of the box.

The aforementioned approaches can be seen as points on a spectrum of design choices.
In this paper, we formalise this spectrum and use the formalization to make principled design choices for a new class of monitors.

\subsection{Our Approach and Contributions}

\paragraph{General problem formulation} Using the theory of sequential testing~\cite{Wald1945-gv,ramdas2023gametheoretic}, we provide a general formulation of the model monitoring problem as testing the hypothesis that the environment is \emph{consistent} with a given \emph{robust model} (\Cref{sec:robust-model}). A robust model assigns, to each history of observations, an \emph{uncertainty set} of distributions over next-step observations. An environment is consistent with the model if, after every history \( h \), its next-step distribution lies in the uncertainty set of \( h \). Apart from this constraint, the environment's distribution for each history \( h \) can be chosen independently of choices for other \( h' \neq h \), i.e., our robust models are \emph{rectangular}, as is usual in the literature~\cite{iyengar,Nilim2005RobustCO,wiessmanrmdp}. We focus on a natural class of \emph{polytopic} rectangular models~\cite{ijcai2024polytopic,aaai25polytopic,concur26polytopic}, where each uncertainty set is a convex polytope. On the other hand, we make no Markov or time-homogeneity assumptions about the environment or the model, and our framework subsumes, e.g., monitoring robust rectangular (partially observable) MDPs.

\paragraph{Betting Monitors: Complete Language for Monitoring of Robust Models} A special class of monitors arises from \emph{test supermartingale} processes in sequential testing theory~\cite{ramdas2023gametheoretic}. These have a natural \emph{betting} interpretation, in which the monitor accumulates wealth by betting against the model and raises an alarm when its wealth becomes sufficiently large. Here, the wealth is a product of one-step \emph{fair bets}: functions assigning numerical payoffs to successor observations whose expectation is at most one under every successor distribution the model admits for the current history~(\Cref{sec:betting-intro}). The idea is that under the model, the bettor cannot expect to increase their wealth, but if the true environment persistently diverges from the model, suitable bets let the bettor's wealth grow unboundedly. For general composite hypotheses (monitoring model ensembles), this approach has limited expressive power: there are cases where no test supermartingale/betting monitor can detect model mismatch, and one must resort to the more general \emph{e-processes}~\cite{howard2021timeuniform,Ramdas2022-jn,ramdas2020admissible,Grunwald2024-pk}. We show that for polytopic robust models, every $\alpha$-valid monitor is \emph{dominated} by an \( \alpha \)-valid betting monitor, i.e., the latter raises an alarm no later than the former \emph{on every trajectory} (\cref{thm:betting-monitor}). Due to this \emph{completeness} result, we can restrict our attention to betting monitors.

\paragraph{Design Choice Spectrum}
Designing a monitor for a polytopic robust model thus reduces to designing a sequential betting algorithm, a process we decompose into four choices (\Cref{sec:betting-algorithms,sec:constructions}): the choice of (1)~\emph{features}, i.e., the summary of the history available to the monitor; (2)~the \emph{parameterisation}, a function mapping the features and a parameter vector to weights expressing how strongly the monitor bets on each possible next observation; (3)~the \emph{bet construction}, turning the weights into a bet that is fair
under the uncertainty set of the current history;
and (4)~the \emph{learning component}, adapting the parameters along the run so as to optimize the bets using the information learned about the true environment. The monitor's validity is ensured by the bet construction alone, so the remaining components can be chosen freely to maximize detection power. The existing monitors discussed above occupy particular points in this design space.

\paragraph{Bet Constructions and their Analysis}
Two standard constructions of fair bets from the optimal betting literature are the \emph{scale} construction, known in the literature on testing by betting~\cite{Ramdas2025-gq}, and the \emph{reverse information projection} (RIPr), which stems from the study of optimal bets against a known distribution~\cite{li1999mixture,Grunwald2024-pk,Lardy2024-gy}. We show how to compute these constructions efficiently in the interval model setting~(\Cref{sec:computation}), while for general polytopic models, we provide new theoretical insights into the RIPr construction: First, RIPr bets can dominate (pointwise) any fair bets (\Cref{thm:ripr-completeness}). Second, we provide a geometric interpretation of the two constructions that shows how RIPr offers more flexibility in producing the bets, which renders it a more suitable choice without other assumptions on the weights parametrisation.

\paragraph{Mixture of Linear Monitors: Algorithm and Pathwise Guarantees}
To propose a concrete monitor design, we instantiate the remaining design choices. The monitors we propose can use arbitrary domain-specific features, e.g., the probabilities predicted by the model or, in a gridworld, the agent's coordinates.
As for parametrisation and learning, we propose a \emph{mixture of linear monitors} (\Cref{sec:algorithm}). Here, the weights are \emph{log-linear} in the features, with a parameter vector $\theta\in\R^d$ used to pool evidence across transitions with similar features. For the learning component, we use the method of mixtures, averaging the wealth processes induced by all $\theta$ using a Gaussian mixture, approximated by a finite number of particles without affecting validity.
We perform a regret and detection time analysis in~\Cref{thm:mixture-regret}  and \Cref{cor:mixture-detection-time} showing that, on every trajectory,
the mixture of linear monitors asymptotically collects the evidence at least as fast as any linear monitor associated with a fixed choice of parameter $\theta^*$.

\paragraph{Experiments} We performed a thorough experimental evaluation that included 925 monitoring instances over 23 environments~(\Cref{sec:experiments}). The results confirmed the correctness of the proposed monitor construction and its practical advantage over existing approaches. In particular, i) the RIPr construction for robust models generally improves the mismatch detection times, ii) our sequential betting formulation instantiated by the mixture of linear monitors outperforms the existing monitors from the literature, and most importantly, iii) it can effectively exploit the observation features, significantly reducing the number of undetected model misalignments.

\section{Setup and Notation}
\label{sec:setup}

{
We write $\ind{\cdot}$ for the indicator function (Iverson bracket), which equals $1$ if its argument holds and $0$ otherwise.
The set of probability distributions on a finite set $\cX$ is the probability simplex
\(
  \dist{\cX}=\{p\in[0,1]^{\cX}:\sum_{x\in\cX}p(x)=1\}.
\)
The support of ${p\in\dist{\cX}}$ is $\supp(p)=\{x\in\cX:p(x)>0\}$, and the
support of a set $\cU\subseteq\dist{\cX}$ is
$\supp(\cU)=\bigcup_{p\in\cU}\supp(p)$. For a function
$f\colon\cX\to\R$ and $p\in\dist{\cX}$, we write
$\E_p[f]$ to denote the expectation of $f$ under $p$.
A \emph{polytope} in $\dist{\cX}$ is a convex hull of finitely many distributions.
For $q,p\in\dist{\cX}$, the
Kullback--Leibler divergence of $q$ from $p$ is
$\KL(q\Vert p)=\sum_{x\in\supp(q)}q(x)\log\frac{q(x)}{p(x)}$
if $\supp(q)\subseteq\supp(p)$, and $\KL(q\Vert p)=\infty$ otherwise.
For a finite set $\cX$, we identify functions $\cX\to\R$ with vectors in
$\R^{|\cX|}$ and compare them component-wise.
For a vector $v\in\R^d$, $\lVert v\rVert_\infty=\max_i|v_i|$. We use
the convention $\inf\varnothing=\infty$, $\frac{0}{0}=0$, and $0 \cdot \infty = 0$.
}

\subsection{Controlled Stochastic Processes}
Let $\cX$ and $\cA$ be finite sets of states and actions, respectively.
The interaction with the environment generates the stochastic process
  $(X_t,A_t)_{t\geq0} = X_0,A_0,X_1,A_1,X_2,\ldots,$
where the $\cX$-valued random variable $X_t$ represents the state observed
at time $t$ and the $\cA$-valued random variable $A_t$ represents the action
taken at time $t$. We call the process  $(X_t,A_t)_{t\geq0}$ the \emph{true process}.
A \emph{history} comprising $t+1$
state-action pairs is a sequence
$
  h_t=(x_0,a_0,x_1,\ldots,x_t,a_t).
$
We use $\cH$ to denote the set of all histories. We further write
$\HH_t=(X_0,A_0,X_1,\ldots,X_t,A_t)$
to denote the random history.
We call a process $(Y_t)_{t\geq0}$ \emph{adapted} if $Y_t$ is a function of $(\HH_{t-1}, X_t)$ for every $t\geq0$, where $\HH_{-1}$ is the empty history.

We make no Markov or time-homogeneity assumption about the true process. Its
one-step transition function after a history $h_t$ is
  $\qstar_{t+1}(x\mid h_t)
  =\Prob(X_{t+1}=x\mid \HH_t=h_t)$ for
$x\in\cX.$
Unless stated otherwise, $\Prob$ and $\E$ refer to probability and
expectation in the probability space underlying the process.

\begin{remark}
For the sake of clarity, we present the theory under the usual MDP interpretation
where $X_t$ is the state of the environment and $A_t$ is the action of the
agent.
All results stand as long as $\cX$ and $\cA$ remain finite, $\cX$ is the quantity predicted by the model, and $\cA$ represents
any other available information. For instance, the same theory applies to Markov chains where $\cA$ becomes empty or to partially observable MDPs, where $\cX$ becomes the space of partial observations.%
\end{remark}

\subsection{Robust Model}
\label{sec:robust-model}

A \emph{robust model} is a function $\model$ that assigns to each prediction
history $h_t \in \cH$ a nonempty set $\model(h_t) \subseteq \dist{\cX}$.  We call the model \emph{simple} if $\model(h_t)$ is a singleton for every history.%
We call $\model(h_t)$ an \emph{uncertainty set}, and use $\cU$ to denote
a generic uncertainty set. This set contains the one-step transition
functions that the model considers plausible for $X_{t+1}$ after $h_t$,
and is the predictive object that is exposed to formal analysis.

We say that a stochastic process is \emph{consistent with the model} if, after
every history that can occur under the process, its conditional one-step
transition function is admitted by the model:
$  \qstar_{t+1}(\cdot \mid h_t)\in\model(h_t)$
  for every $h_t\in\cH$ with $\Prob(\HH_t=h_t)>0.$
  Thus, a process consistent with the model may use a different admissible
transition function at every time and after every history; it need not be
Markov or time-homogeneous.
We say a history $h_t$ \emph{is consistent with $\model$} if $\Prob(\HH_t = h_t) > 0$ for at least one process consistent with $\model$.

We focus on \emph{polytopic robust models}~\cite{ijcai2024polytopic}, for which
$\model(h_t)$ is a convex polytope for every history $h_t$. Some results require
the more specific class of \emph{interval robust models}~\cite{GIVANinterval}: for every $h_t$,
there are lower and upper bounds $a,b\in[0,1]^\cX$ with $a\leq b$ such that
  $\model(h_t) = \dist{\cX} \cap \prod_{x\in\cX}[a_x, b_x].$ In particular, our definition subsumes rectangular robust MDPs~\cite{iyengar}.
\begin{example}
\label{ex:line-world}
We consider a random walk in a one-dimensional windy gridworld as our running example, illustrated in~\cref{fig:line-world} in Appendix~\ref{sec:examples}.
The states are the cells $\cX=\{1,\dots,N\}$ and the actions are $\cA=\{-1,+1\}$ (left/right).
From cell $x$, the agent either reaches the intended cell $x+a$ or slips to the opposite neighbour $x-a$; a move leaving the line keeps the agent in place; the intended move succeeds with probability $p(x+a\mid x,a)=0.75$ for both actions.
The robust model $\model$ is an interval MDP that widens each transition probability $p$ to $[p-0.1, p+0.1]$, capturing uncertainty in the slip probability.
In the real environment, the wind shifts probability $0.2$ to the right neighbour, so the intended move succeeds with probability $0.95$ for $a=+1$ and $0.55$ for $a=-1$.
Both probabilities lie outside the interval $[0.65,0.85]$ admitted by $\model$, so the real environment is not consistent with $\model$.
\end{example}

\subsection{Problem Statement}
A \emph{monitor} is any adapted $\{0,1\}$-valued process
$(\monitor_t)_{t\geq0}$. The value $\monitor_t=0$ means that the observations
up to time $t$ are considered consistent with the model, while $\monitor_t=1$ raises
an alarm. We denote with $\tau_\nu$ the time of the first alarm:
$\tau_\nu=\inf\{t\geq0:\monitor_t=1\}.$

For $\alpha\in(0,1)$, the monitor is \emph{$\alpha$-valid} if its probability
of ever raising a false alarm is uniformly bounded:
$  \Prob(\tau_\nu<\infty)\leq\alpha$
  under every process consistent with $\model.$
Among monitors satisfying this guarantee, we seek ones that raise an alarm
quickly when the true process is not consistent with the model.

Without loss of generality, we assume the mild support condition:
\(
  \supp(\qstar_{t+1})\subseteq \supp(\model(\HH_t))
\)
on every history consistent with $\model$.
Indeed, if this does not hold, we may observe outcomes that are impossible under every distribution
admitted by the model $\delta$, in which case the monitor would trivially raise an alarm.

\section{Betting Monitors}
\label{sec:betting-intro}
A special class of monitors arises from \emph{test supermartingales}~\cite{ville1939etude,Shafer2011-xg,Ramdas2025-gq}. We view these through the lens of \emph{sequential betting.} Such a monitor can be viewed as a bettor who in every timestep $t$ possesses some wealth $W_t$, starting with $W_0=1$. We call $(W_t)_{t\geq0}$ a \emph{wealth process} induced by the bettor. At every timestep $t$, the monitor and model set up a bet. The model believes that the dynamics are consistent with
$\model$ and accepts any bet that is fair under every distribution it
considers plausible. The bettor bets against the model and aims to
accumulate as much wealth as possible. Formally, at each $t$, the bettor stakes
their entire wealth $W_t$ and announces a payoff function
$e\colon\cX\to[0,\infty]$. If the next state is
$X_{t+1}=x$, the bettor receives $W_t \cdot e(x)$, so that the new
wealth is $W_{t+1}=W_t \cdot e(X_{t+1})$. The model accepts a bet only
if it cannot lose money in expectation under any distribution plausible according to $\cU = \model(h_t)$,
i.e., if
\vspace{-.5em}
\begin{equation}
\label{eq:fair-bet}
\E_{x\sim p}[e(x)]\leq1
\vspace{-.5em}
  \qquad\text{for every }p\in\cU.
\end{equation}
We call any such function $e\colon\cX\to[0,\infty]$ satisfying \eqref{eq:fair-bet}
a \emph{fair} (or \emph{valid}) bet for $\cU$.\footnote{In the statistics literature, functions satisfying \eqref{eq:fair-bet} are known as \emph{e-values}~\cite{vovk2021evalues}.}
With the initial budget $W_0=1$, and repeating this procedure with a fair bet
$e(\,\cdot\,;h_{i-1})$ for $\model(h_{i-1})$ in every step $i$,
the bettor's wealth after $t$ steps is
  $W_t = \prod_{i=1}^{t} e(X_{i};\HH_{i-1})$,
where each bet is chosen based on the
history $h_{i-1}$. If the true process is
consistent with $\model$, the bettor cannot expect to increase their wealth.
If, however, the true dynamics differ from the model, the bettor may choose
payoffs whose expectation is at most one under every model distribution but
ensure the wealth grows under the distribution they believe governs the~environment.

\begin{example}
\label{ex:fair-bet}
In the windy line world of \cref{ex:line-world}, under the action $a=+1$,
the model predicts that the next state is $x+1$ with probability in $[0.65,0.85]$ and $x-1$ with probability in $[0.15,0.35]$.
The bet $e_1(x+1)=1.1$ and $e_1(x-1)=0.4$ is fair, because its maximum expectation under the model is $0.85\cdot1.1 + 0.15\cdot0.4 = 0.995 < 1$.
\end{example}

Formally, a \emph{betting monitor} for a wealth process $W_t$ is the monitor $\monitor_t = \ind{W_t \geq 1/\alpha}$ that raises an alarm in step $t$ if $W_t \geq \frac{1}{\alpha}$, where $\alpha$ is the given validity parameter. We refer the reader to Appendix~\ref{app:test-sm} for deeper background on test supermartingales. The following is a simple consequence of Ville's inequality~\cite{ville1939etude} and of the fact that a betting monitor must produce fair bets:

\begin{theorem}
\label{thm:betting-validity} Every betting monitor is \(\alpha\)-valid.
\end{theorem}

However, it is generally not the case that betting monitors are good at detecting model misalignment. Indeed, in general sequential
hypothesis testing, there are settings in which test supermartingales/betting monitors cannot detect model mismatch and the appropriate tool is a more general class of processes called e-processes
\cite{howard2021timeuniform,Ramdas2022-jn,ramdas2020admissible,Grunwald2024-pk}. In~\Cref{thm:betting-monitor} below (proven in the appendix) we show that for robust polytopic models, this does not happen: betting monitors are \emph{complete} for monitoring of such models, in the sense that
every valid monitor is dominated by a monitor based on a test supermartingale.

\begin{restatable}{theorem}{thmBettingMonitor}
\label{thm:betting-monitor}
Let $(\monitor_t)_{t\geq0}$ be an $\alpha$-valid monitor for a polytopic robust model $\model$.
Then there is a betting monitor $(\hat{\monitor}_t)_{t\geq 0}$
for $\model$
which raises its first alarm no
later than $(\monitor_t)_{t\geq0}$ on every history consistent with $\model$.
\end{restatable}

\subsection{Power of Bets}
\label{sec:power}
Validity alone says nothing about detection; for instance, the constant bet $e(x)=1$ for $x\in \cX$ is fair, but
it never increases the wealth. To measure how quickly a bet gathers evidence, we consider the log-wealth, which expresses the product of payoffs (i.e., the accumulated evidence against the model) as a sum:
$\log W_t=\sum_{i=1}^t\log e(X_i;\HH_{i-1})$. If the bettor
repeatedly places a bet $e$ against outcomes drawn independently from a distribution $q$,
then $\frac1t\log W_t\to\E_q[\log e]$ almost surely by the law of large
numbers. From the long-run perspective, it thus makes sense to look at the value $\log W_t$
rather than $W_t$, as $\log W_t$ eventually concentrates around $t\cdot \E_q[\log e]$.
We call $\E_q[\log e]$ the \emph{growth rate} of $e$ under $q$. If it
is positive, the wealth grows like $\exp(t\,\E_q[\log e])$ and crosses
$1/\alpha$ after about $\log(1/\alpha)/\E_q[\log e]$ steps; if it is
negative, the wealth vanishes.

\begin{excont}{ex:fair-bet}
Under the windy dynamics, the bet $e_1$ has growth rate of
$0.95\log(1.1)+0.05\log(0.4)\approx0.045>0$, so, provided the agent always selects the action $a=+1$,
a bettor who always places $e_1$ reaches wealth $1/\alpha=20$ after about $\log(20)/0.045\approx67$ steps.
Also note that requiring $\E_q[e]>1$ is not enough for a consistent growth of the evidence; the fair bet $e_2(x+1)=1/0.85$, $e_2(x-1)=0$ has expected
payoff under the windy dynamics $0.95/0.85\approx1.12$, but it loses
all wealth on the first occurrence of a slip that happens almost surely.
\end{excont}

\subsection{Design Choices for Betting Monitors}
\label{sec:betting-algorithms}

By \Cref{thm:betting-monitor}, restricting attention to betting monitors loses no monitoring power. How quickly the alarm comes depends on the bets. The wealth grows fastest when
the monitor bets on outcomes that the real environment makes more likely than
the model admits. Since the real environment is unknown, the monitor has to
learn such bets from the history, and its performance is governed by two
design choices: 1) what information about the history the betting rule uses, and 2) how rich a family of betting rules the monitor can represent. Rich families that utilize extensive information can represent better bets,
but it takes the monitor longer to learn which bets are effective.
Small, less informed families are easier to optimize over,
but the bets they contain may be weaker.

\begin{example}
\label{ex:tradeoff}
We illustrate this trade-off on the windy line world of
\cref{ex:line-world} with $N=13$, $\alpha=0.05$, and an agent that always picks the action ``left'' with probability $0.9$. In every cell and for both actions,
the real environment makes the right neighbour more likely than the model
admits, so a good monitor should learn to bet more on the right neighbour.
In this specific scenario, the growth-optimal monitor tailored for the real windy
dynamics~\cite{Grunwald2024-pk}\footnote{This monitor serves only as a reference: it ``cheats'' by knowing the true dynamics in advance.} raises the alarm after about $130$ steps on average.

\emph{Transition-wise bets.} The sequential tests for Markov chains
of~\cite{Fields2025-gu,sethi2026markovian} place their bets based on the current state $x$, the
action $a$, and the successor $x'$. They estimate the transition kernel
$\hat{q}(\,\cdot\mid x,a)$ from the history and place the bet that would grow
the wealth fastest if the estimate were correct. As $\hat{q}$ converges to the
true kernel $\qstar$, the bets converge to the optimal ones. However, the monitor
needs to learn one parameter per transition, i.e., separately in every cell, although the wind is the same everywhere.
In practice, it thus takes much longer than the optimal $130$ steps to identify accurate parameters and good bets, and in our evaluation,
the transition-wise monitor \texttt{TW} (\cref{sec:experiments}) raises the alarm after about $2100$ steps.

\emph{Score-based bets.} The monitors of \citet{henzinger2025alignment} bet
depending solely on how well each outcome agrees with the model's prediction, i.e.,
$e(x';h_t)\propto f_\theta(p,p(x'))$, where $p$ is the prediction of
$\model(h_t)$\footnote{The monitors were designed for simple models, so $\model(h_t)$ is a singleton.}
and $f_\theta$ is a function parametrised by $\theta$ that is invariant to permutations of the outcomes: it sees the predicted
probabilities, but not the outcomes they belong to. In the case of the line world, the function thus cannot distinguish between the actions
as in both cases, the appearing values of probabilities are the same.
Such a monitor therefore cannot directly bet on the outcome corresponding to the right neighbour.
Still, it can gain evidence on
average by betting on the less likely successor, which, in $90\%$ of the cases, is the right neighbour.
We made the bets valid for the robust model by scaling them according to the scale construction of \cref{sec:constructions}.
With only a few parameters,
the monitor quickly learns its best bet: the scoring rule-based monitor
\texttt{SR-Br-Scale} (\cref{sec:experiments}) raises the alarm after about
$650$ steps with the best fixed parameter and after about $1200$ steps when
learning it. Its best bets are weaker than the transition-wise ones, but it finds
them much faster. The specific choice of a scoring rule and other components also have a cost:
the best bets of the form $f_\theta$ could raise the alarm after about $380$ steps,
so a better bet representation could improve the performance even with the same information.
In the line world used in the experiments of \cref{sec:experiments} the agent picks each of the actions with probability one half. Although the true transition probabilities differ from the model's predictions for both actions, these discrepancies cancel out when aggregated across actions: outcomes assigned probability $0.25$ by the model occur with that same frequency overall, and the same for the other outcome. Consequently, bets of the form $f_\theta$ cannot detect the misalignment, consistent with our experimental results.
\end{example}

The two designs sit at opposite ends of the trade-off: transition-wise bets use the full
history but learn one parameter per transition, while score-based bets learn only a few
parameters overall but see a very limited observation of the states and actions involved.
To capture both in one framework, we make the two design choices explicit.
The information is given by a \emph{feature map}: the monitor observes a summary of the
history\footnote{The history summary may be chosen by the user or implied by the way the model was obtained.}
$  \feat_t=\feat(h_t)\in\R^{|\cX|\times d}$,
whose $x$-th row $\feat_{t,x}\in\R^d$ describes the transition to the outcome~$x$.

The family of bets is given by a \emph{parametrised betting rule}: every parameter $\theta$
from a parameter space $\Theta$ defines a payoff
\vspace{-.5em}
\[
  e^\theta(x;\feat_t)
  \qquad\text{to every possible outcome }x\in\cX,
  \vspace{-.5em}
\]
which must be a fair bet for $\model(h_t)$ for every feature matrix and every parameter.
The two designs above are the two ends of this spectrum. Transition-wise bets use the
indicators of the current state, the action, and the successor as features, with one
parameter per transition; score-based bets use the model's prediction as features, with
only a few parameters (concretely two). In between, the features
and the family can be chosen independently, including features that neither design uses,
which can make the trade-off much more favourable. The parameter $\theta$ itself may be
fixed in advance, fitted from previous data, or aggregated over as in the mixture
construction of \cref{sec:algorithm}.

\begin{excont}{ex:tradeoff}
\emph{Coordinate-based bets.} In the line world, a natural choice is to use the
coordinates of the current state and the successor as features, giving bets
$e(x';h_t)\propto f_\theta(p(x'),x',x)$. The real kernel is a function of these
features, $\qstar(x'\mid x,a)=p(x')+0.2(x'-x)$, so such a monitor can in
principle place the optimal bets with mean detection time $130$ steps. The
coordinate-based monitor \texttt{ML-RIPr-C} (\cref{sec:experiments}) raises
the alarm after about $175$ steps with the best fixed parameter and after about
$350$ steps when learning it, much faster than the other monitors thanks to
its low-dimensional parametrization and informative features. Moreover, its
best bets are close to the best ones attainable with these features, which we attribute to
the flexibility of the linear model compared to a fixed scoring rule.
\end{excont}

\section{Bet Constructions}
\label{sec:constructions}
A crucial step in designing a betting monitor is to design a parameterised betting rule that
produces a fair bet for every feature matrix $\feat_t$ and every parameter
value $\theta$. A general approach is to first map every
such pair into an $\cX$-indexed vector of \emph{weights}
\(
  w_x(\feat_t,\theta) > 0, x\in\cX,
\)
that express how strongly the monitor prefers to bet on each observation $x$ given the current history (encoded by $\feat_t$);
and then transform the weights into a valid bet for $\model(h_t)$
by a bet construction that guarantees validity for any input~weight.

Formally, a \emph{bet construction} $\cC$ assigns to every nonempty polytope
$\cU\subseteq\dist{\cX}$ and every weight vector $w\in(0,\infty)^{\cX}$
a bet
\(
  \cC(w,\cU)\in(0,\infty]^{\cX}
\)
that is \emph{fair} for $\cU$ and satisfies $\cC(w,\cU)(x)=\infty$ for
every $x\notin\supp(\cU)$.
Given a construction $\cC$ and a weight vector parametrisation $w(\cdot,\theta)$, the
monitor bets with
\vspace{-.5em}
\[
e^\theta(\,\cdot\,;h_t)=\cC\bigl(w(\feat_t,\theta),\model(h_t)\bigr).
\vspace{-.5em}
\]
By definition, this bet satisfies \eqref{eq:fair-bet} for every
history and every parameter. The function producing the weights can thus be chosen freely --- it could be, e.g., a regression tree, a neural network, or a linear model ---
and the validity of the resulting monitor is guaranteed by the transformation alone.

We propose and compare two bet constructions, the \emph{scale construction} (Section~\ref{sec:scale}) and the \emph{RIPr construction} (Section~\ref{sec:ripr}).%

\subsection{Scale Construction}
\label{sec:scale}
The first construction, appearing in the literature on testing by betting
(e.g. proof of Theorem~3.17 in \cite{Ramdas2025-gq}), rescales the weights by a
constant factor. In this case, the weights express the relative stakes of the bet, i.e., if
$w_x=2w_y$, the monitor bets twice as much on the outcome $x$ as on $y$.
The \emph{scale construction} $\cC^{\mathrm{Sc}}$ is given by
\[
  \cC^{\mathrm{Sc}}(w,\cU)(x)=\dfrac{w_x}{\max_{p\in\cU}\E_p[w]}.
\]
The expectation of the resulting bet $e$ is maximised at $p^*=\argmax_{p\in \cU}\E_p[e]$, where it satisfies $\E_{p^*}[e]=1$.
Therefore $\cC^c$ is indeed a valid bet construction.

\subsection{RIPr Construction}
\label{sec:ripr}

The second construction is based on the notion known in the betting literature as reverse information projection~\cite{Lardy2024-gy} defined as follows.
\begin{definition}
Let $\cU\subseteq\dist{\cX}$ be a nonempty uncertainty set and let
$q\in\dist{\cX}$. The \emph{reverse information projection} (RIPr) of $q$
onto $\cU$ is the set of minimisers
\begin{equation}
\label{eq:ripr}
  \RIPr(q,\cU)=
  \argmin_{p\in\cU}\KL(q\Vert p).
\end{equation}
We also abuse the notation and use the term \emph{RIPr} for any member of $\RIPr(q,\cU)$.
\end{definition}

It is a well known result~\cite{li1999mixture,Lardy2024-gy,Grunwald2024-pk} that
whenever a RIPr $p^*$ exists for $q$, the likelihood ratio $q/p^*$ is a fair bet for $\cU$. The \emph{RIPr construction} is then given by
\[
  \cC^{\mathrm{RIPr}}(w,\cU)(x)=q(x)/p^*(x)
  \qquad  \text{ for } \qquad p^*\in\RIPr(q,\cU),
\]
where $q(x)\propto\ind{x\in\supp(\cU)} w_x$
is a distribution with $\supp(q)=\supp(\cU)$, for a positive weight vector $w$.

If $\cU$ is a polytope, the minimum is attained because $\cU$ is compact and
$\KL(q\Vert\,\cdot\,)$ is lower semicontinuous, so $\RIPr(q,\cU)$ is
nonempty. In general it may contain more than one distribution; the
likelihood ratio between $q$ and its RIPr is nevertheless unique.
We formalise the discussion above as follows:
\begin{restatable}{theorem}{thmRiprWellDefined}
\label{thm:ripr-well-defined}
The $\cC^{\mathrm{RIPr}}$ is a well-defined bet construction.
\end{restatable}

Ideally, a construction should not limit the bettor: every fair bet the bettor
might want to place should be obtainable from some weights, or at least be
improved upon by such a bet. It is not obvious that the RIPr
construction has this property, since it only produces likelihood ratios $q/p^*$
between a distribution and its projection. While this is indeed not true in general, we prove the result for polytopic $\cU$.

\begin{restatable}{theorem}{thmRiprCompleteness}
\label{thm:ripr-completeness}
Let $\cU$ be a nonempty polytope. For every fair bet $e$ for $\cU$ there exists $e'$ with $e \leq e'$ such that $e' = q/p^*$ for $q$ satisfying $\supp(q)\subseteq\supp(\cU)$ and $p^* \in \RIPr(q, \cU)$.
\end{restatable}

\subsection{Computation}
\label{sec:computation}
In the general case of a polytopic $\cU$, the scale construction is easier to compute than the RIPr construction: in the former, the scale constant is a solution of a linear program, while computing the RIPr requires solving a convex optimisation problem.
However, for interval robust models, both constructions can be computed exactly
and efficiently in time $\mathcal O(n\log n)$, where $n=|\supp(\model(h_t))|$.
We briefly describe the algorithms for computing the scale factor and the RIPr.

\paragraph{\textbf{Scale Construction}}
\label{sec:scale-computation}
For interval robust models, the linear program
$\max_{p\in\model(h_t)}\E_p[w]$ is a fractional knapsack, solved greedily: start from the lower bounds $p=a$ and distribute the remaining
mass $1-\sum_x a_x$ over the outcomes in decreasing order of $w_x$, raising
each $p(x)$ up to $b_x$. The greedy approach requires sorting the outcomes by $w_x$, which yields the $\mathcal O(n\log n)$ complexity, where $n=|\supp(\model(h_t))|$.

\paragraph{\textbf{RIPr Construction}}
\label{sec:ripr-computation}
For interval models, the RIPr is a clipped rescaling of $q$.

\begin{restatable}[Thm. 1 of \cite{DAmico2014-zg}]{lemma}{lemIntervalRipr}
\label{lem:interval-ripr}
Let $\cU=\dist{\cX}\cap\prod_{x}[a_x,b_x]$ be a nonempty interval set and
let $q\in\dist{\cX}$ satisfy $\supp(q)=\supp(\cU)$. Then there is $\mu>0$
such that $\RIPr(q,\cU)=\{p^*\}$ with
\begin{equation}
  \label{eq:interval-ripr}
  p^*(x)=\min\left\{\max\left\{q(x)/\mu,\,a_x\right\},\,b_x\right\}
  \qquad\text{for every }x\in\cX.
\end{equation}
\end{restatable}
Thus, computing the RIPr amounts to finding $\mu$ that makes $p^*$ a distribution.
This is solvable by a sweep line algorithm as described in Algorithm~1 of \citet{DAmico2014-zg}.
In our case, the algorithm runs in time $\mathcal O(n\log n)$.

\subsection{Bet Geometry}
\label{sec:bet-geometry}

For the scale bet, the level sets of $p \mapsto \E_p[w]$ are parallel hyperplanes, so for weights in general position it is maximised at a single vertex
(\cref{fig:bet-geometry}, left).
A larger face requires exact values of $e$, which the scale bet with weights in general position does not~produce.

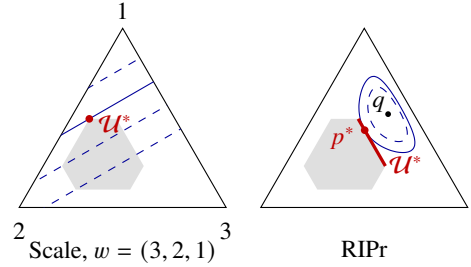
\begin{wrapfigure}[17]{r}{0.52\textwidth}
  \centering
  \vspace{-2.5em}
  \def\bary#1#2#3{({4.0*(0.5*#1+#3)},{3.464*#1})}
  \def\hexagon{\bary{.5}{.3}{.2} -- \bary{.5}{.2}{.3} -- \bary{.3}{.2}{.5}
    -- \bary{.2}{.3}{.5} -- \bary{.2}{.5}{.3} -- \bary{.3}{.5}{.2} -- cycle}
\resizebox{\linewidth}{!}{%
  \begin{tikzpicture}[font=\small, scale=0.88, transform shape]
    \begin{scope}
      \begin{scope}[shift={(-0.77cm,-0.43cm)}]
        \fill[gray!25] \hexagon;
      \end{scope}
      \begin{scope}[scale=0.76]
        \draw \bary100 -- \bary010 -- \bary001 -- cycle;
        \node[above, font=\footnotesize, transform shape=false] at \bary100 {$1$};
        \node[below=2pt, font=\footnotesize, transform shape=false] at \bary010 {$2$};
        \node[below=2pt, font=\footnotesize, transform shape=false] at \bary001 {$3$};
      \end{scope}
      \begin{scope}
        \begin{scope}[scale=0.76]
          \clip \bary100 -- \bary010 -- \bary001 -- cycle;
          \begin{scope}[scale=1.3158]
            \begin{scope}[shift={(-0.77cm,-0.43cm)}]
              \foreach \c in {0.9,1.1,1.5}
                \draw[blue!60!black, dashed] \bary{-1}{\c+2}{-\c} -- \bary{1.5}{\c-3}{2.5-\c};
              \draw[blue!60!black] \bary{-1}{3.3}{-1.3} -- \bary{1.5}{-1.7}{1.2};
            \end{scope}
          \end{scope}
        \end{scope}
      \end{scope}
      \begin{scope}[shift={(-0.77cm,-0.43cm)}]
        \fill[red!75!black] \bary{.5}{.3}{.2} circle (1.6pt);
        \node[red!75!black, below right, yshift=5pt, font=\small, transform shape=false] at \bary{.5}{.3}{.2} {$\cU^*$};
      \end{scope}
      \begin{scope}[scale=0.76]
        \node[below=3mm, font=\footnotesize, transform shape=false] at \bary{0}{.5}{.5} {Scale, $w=(3,2,1)$};
      \end{scope}
    \end{scope}
    \begin{scope}[xshift=3.55cm]
      \begin{scope}[shift={(-0.77cm,-0.43cm)}]
        \fill[gray!25] \hexagon;
      \end{scope}
      \begin{scope}[scale=0.76]
        \draw \bary100 -- \bary010 -- \bary001 -- cycle;
        \node[above, font=\footnotesize, transform shape=false] at \bary100 {};
        \node[below=2pt, font=\footnotesize, transform shape=false] at \bary010 {};
        \node[below=2pt, font=\footnotesize, transform shape=false] at \bary001 {};
      \end{scope}
      \begin{scope}[shift={(-0.77cm,-0.43cm)}]
        \draw[blue!60!black, dashed] plot[smooth cycle, tension=0.3] coordinates {(2.640,2.099) (2.617,2.124) (2.591,2.149) (2.561,2.173) (2.528,2.191) (2.495,2.200) (2.464,2.197) (2.436,2.184) (2.414,2.164) (2.395,2.139) (2.380,2.111) (2.368,2.083) (2.359,2.054) (2.353,2.026) (2.348,1.998) (2.345,1.972) (2.343,1.946) (2.342,1.922) (2.342,1.898) (2.344,1.875) (2.346,1.853) (2.349,1.832) (2.352,1.811) (2.356,1.791) (2.361,1.772) (2.367,1.753) (2.373,1.735) (2.380,1.717) (2.387,1.699) (2.395,1.682) (2.403,1.665) (2.413,1.648) (2.423,1.631) (2.433,1.615) (2.445,1.599) (2.457,1.583) (2.470,1.567) (2.484,1.552) (2.499,1.536) (2.515,1.521) (2.533,1.506) (2.551,1.491) (2.571,1.477) (2.592,1.463) (2.615,1.450) (2.640,1.437) (2.666,1.426) (2.694,1.415) (2.724,1.407) (2.755,1.401) (2.787,1.398) (2.819,1.399) (2.850,1.406) (2.879,1.420) (2.901,1.442) (2.916,1.473) (2.922,1.510) (2.920,1.549) (2.913,1.588) (2.904,1.623) (2.894,1.655) (2.884,1.682) (2.874,1.707) (2.865,1.728) (2.857,1.747) (2.849,1.764) (2.842,1.780) (2.835,1.794) (2.828,1.808) (2.821,1.820) (2.815,1.832) (2.809,1.843) (2.803,1.854) (2.798,1.865) (2.792,1.875) (2.786,1.886) (2.780,1.896) (2.774,1.906) (2.768,1.917) (2.762,1.927) (2.755,1.938) (2.748,1.950) (2.740,1.962) (2.732,1.975) (2.723,1.988) (2.713,2.003) (2.702,2.019) (2.690,2.036) (2.676,2.055) (2.659,2.076)};
        \draw[blue!60!black] plot[smooth cycle, tension=0.3] coordinates {(2.640,2.189) (2.610,2.227) (2.574,2.268) (2.532,2.308) (2.485,2.340) (2.438,2.356) (2.395,2.351) (2.358,2.331) (2.328,2.301) (2.303,2.266) (2.282,2.228) (2.265,2.190) (2.251,2.151) (2.240,2.114) (2.231,2.077) (2.224,2.042) (2.219,2.007) (2.215,1.973) (2.212,1.940) (2.211,1.908) (2.212,1.877) (2.213,1.846) (2.216,1.816) (2.220,1.787) (2.225,1.758) (2.231,1.729) (2.238,1.701) (2.247,1.674) (2.257,1.647) (2.268,1.620) (2.281,1.594) (2.295,1.568) (2.310,1.543) (2.326,1.519) (2.344,1.495) (2.363,1.472) (2.384,1.449) (2.406,1.427) (2.430,1.406) (2.455,1.386) (2.482,1.366) (2.510,1.347) (2.540,1.329) (2.571,1.313) (2.605,1.297) (2.640,1.282) (2.677,1.268) (2.717,1.257) (2.758,1.247) (2.801,1.240) (2.846,1.237) (2.891,1.238) (2.935,1.246) (2.976,1.264) (3.009,1.294) (3.029,1.338) (3.033,1.394) (3.025,1.455) (3.009,1.513) (2.991,1.565) (2.973,1.609) (2.956,1.647) (2.941,1.680) (2.928,1.708) (2.916,1.733) (2.905,1.755) (2.894,1.775) (2.885,1.793) (2.876,1.810) (2.868,1.825) (2.860,1.840) (2.852,1.854) (2.844,1.868) (2.837,1.881) (2.830,1.894) (2.822,1.907) (2.815,1.919) (2.808,1.932) (2.800,1.945) (2.792,1.959) (2.784,1.973) (2.775,1.987) (2.766,2.002) (2.756,2.019) (2.745,2.036) (2.732,2.055) (2.719,2.076) (2.703,2.100) (2.686,2.126) (2.665,2.156)};
        \draw[red!75!black, very thick] \bary{.5}{.2}{.3} -- \bary{.3}{.2}{.5};
        \fill \bary{.52}{.08}{.40} circle (1.2pt);
        \node[font=\small, transform shape=false] at (2.504,1.942) {$q$};
        \fill[red!75!black] \bary{0.452}{.2}{0.348} circle (1.6pt);
        \node[red!75!black, left, yshift=-3pt, font=\small, transform shape=false] at \bary{0.452}{.2}{0.348} {$p^*$};
        \node[red!75!black, right, yshift=-3pt, font=\small, transform shape=false] at \bary{.34}{.2}{.46} {$\cU^*$};
      \end{scope}
      \begin{scope}[scale=0.76]
        \node[below=3mm, font=\footnotesize, transform shape=false] at \bary{0}{.5}{.5} {RIPr};
      \end{scope}
    \end{scope}
  \end{tikzpicture}
}
  \setlength{\abovecaptionskip}{3pt}
  \caption{Fair-bet maximisers over the shaded polytope $\cU$. Left: for generic weights, the scale bet is maximised at the single vertex $\cU^*$. Right: the RIPr $p^*$ for $q$ is where the level sets $\KL(q\Vert\,\cdot\,)$ first touch $\cU$ (solid); $q/p^*$ is maximised on the whole red edge.}
  \label{fig:bet-geometry}
\end{wrapfigure}

The expectation of the RIPr bet $q/p^*$ is maximised in $p^*$ where the level sets of $\KL(q\Vert\,\cdot\,)$ around a distribution $q$ first touch $\cU$, and therefore also on the whole face containing $p^*$ (as since $p\mapsto\E_p[e]$ is linear, the distributions in $\cU$ maximising the expectation of a bet form a face of $\cU$), see
\cref{fig:bet-geometry}, right.
Such a flexibility in choosing the resulting bet is important.
For instance, in the case of \(\qstar\in\cU\), the constant bet $e(x)=1$ (for all $x$) is maximised on all of $\cU$, and it is the only good bet in that case; every other fair bet $e'$ strictly loses evidence (\Cref{sec:power}), since by Jensen's inequality and fairness, $\E_{\qstar}[\log e']<\log\E_{\qstar}[e']\leq0$.
The RIPr construction produces $e(x)=1$ from every $q\in\cU$ (since the RIPr onto $\cU$ acts as an identity inside $\cU$), whereas the scale construction produces it only from a constant weight vector $w_x = c$ (for all $x$).

\section{Mixture of Linear Monitors}
\label{sec:algorithm}
We now present a concrete monitor that leverages the bet constructions of \cref{sec:constructions}. This monitor is constructed via a mixture of betting rules, where each bet uses weights $w_x(\feat_t,\theta)$ generated via a linear function. We also establish an important regret result: the evidence gap (or difference in accumulated wealth) between the best fixed-parameter monitor and our mixture monitor is logarithmic in the number of time steps. This also implies only a logarithmic detection delay.

\paragraph{Linear monitors} Let
$\feat\in\R^{|\cX|\times d}$ be a feature matrix whose $x$-th
row $\feat_{x}\in\R^d$ contains features of the transition to the outcome
$x$.
For a parameter $\theta\in\R^d$ we define log-linear weights
\begin{equation}
  \label{eq:linear-weights}
  w_x(\feat,\theta)
  =\exp(\feat_{x}\theta),
  \qquad x\in\cX.
\end{equation}
The parameter vector $\theta$ is shared by all outcomes $x\in\cX$, so a low-dimensional
$\theta$ can describe a rich family of bets, and evidence is pooled across
all transitions with similar features.

\begin{example}
In the windy line world of \cref{ex:line-world}, let $p(x)$ be the
probability of the successor $x$ predicted by the model. In our experiments,
the monitor \texttt{ML} uses the prediction-based features
$\feat_{x}=(\log p(x),\,p(x),\,p(x)^2)$, and the monitor \texttt{ML-C}
uses $\feat_{x}=(\log p(x),\,c(x))$, where $c(x)$ is the coordinate of the
successor $x$.
\end{example}

The monitor turns these weights into a bet with either construction of
\cref{sec:constructions}, i.e.,
$e^\theta(\,\cdot\,;h_t)=\cC(w(\feat_t,\theta),\model(h_t))$ for
$\cC = \cC^{\mathrm{Sc}}$ or $\cC = \cC^{\mathrm{RIPr}}$.
For either bet, this parameterised linear monitor accumulates the wealth according to the wealth process
\(
  W^\theta_t=\prod_{s=1}^{t}e^\theta(X_s;\HH_{s-1}).
\)

\paragraph{Mixing over parameters}
The parameter $\theta$ that extracts the most evidence is unknown in advance and
one could replace it by a sequence $(\theta_t)_{t\geq1}$, where $\theta_{t+1}$ is fitted
to $h_t$ by online learning. Instead, we use the method of mixtures \cite{Wald1945-gv,Robbins1970-qa,choe2021comparing}, which
simultaneously tracks a weighted combination of several wealth processes of the form $(W^\theta_t)_{t\geq0}$.
For any parameter distribution $\pi$ on $\R^d$, we define the wealth of the \emph{mixture of linear monitors} as
\begin{equation}
\vspace{-.6em}
\label{eq:mixture-wealth}
  W^{\LM}_t=\int_{\R^d} \pi(\theta)W^\theta_t\ \ \mathrm{d}\theta.
  \vspace{-.1em}
\end{equation}
In particular, for finite support distribution we obtain a
convex combination $\sum_{k=1}^K\pi_k\,W^{\theta_k}_t$ of $K$ wealth processes.
In Appendix~\ref{app:mixtures} we show that the process $W^\LM_t$ is a wealth process of  a single bettor. Consequently, we show that the method of mixtures can be seen as learning the best parameter through a distributional
rather than a point estimate.

\paragraph{Our algorithm}
\Cref{alg:mixture-monitor} implements the mixture monitor for interval
robust models. For a continuous distribution
$\pi$, the integral \eqref{eq:mixture-wealth} cannot be computed exactly.
The algorithm therefore draws $K$ particles $\theta_1,\ldots,\theta_K\sim\pi$
once, before monitoring starts, and sets their log-wealth to zero
(line~\ref{alg:line:init}). It then tracks the finite mixture with uniform
weights $\widehat W^\LM_t=\frac1K\sum_{k=1}^K W^{\theta_k}_t$.

Each step costs $\mathcal O(Kn(d+\log n))$ time for
$n=|\supp(\model(h_t))|$, with either construction.
In \cref{sec:regret}, we analyse the exact mixture
\eqref{eq:mixture-wealth} for a Gaussian distribution $\pi$; in our
experiments, we use its particle approximation as described above, with
$K=300$ particles drawn from the standard Gaussian distribution
(\cref{sec:experiments}).

\begin{algorithm}[t]
\caption{Mixture of linear monitors}
\label{alg:mixture-monitor}
\begin{algorithmic}[1]
\Require Interval robust model $\model$, feature map $\feat$, distribution $\pi$ on $\R^d$,
         number of particles $K$, level $\alpha$, construction (scale or RIPr)
\State Sample $\theta_1,\ldots,\theta_K\sim\pi$; \quad $\log W^k\gets0$ for $k=1,\ldots,K$ \label{alg:line:init}
\For{$t=0,1,2,\ldots$}
  \State Observe $h_t$; obtain bounds $a,b$ of $\model(h_t)$; compute $\feat_t=\feat(h_t)$ \label{alg:line:observe}
  \For{$k=1,\ldots,K$}
    \State $w^k\gets w(\feat_t,\theta_k)$ by
           \eqref{eq:linear-weights} \label{alg:line:weights}
    \State $e^k\gets \cC(w^k, \model(h_t))$ with $\cC \in \{\cC^\mathrm{Sc}, \cC^\mathrm{RIPr}\}$
      \Comment{\cref{sec:computation}} \label{alg:line:bet}
  \EndFor
  \State Observe $X_{t+1}$ \label{alg:line:outcome}
  \State \textbf{if} $X_{t+1}\notin \supp(\model(h_t))$ \textbf{then} raise an alarm  \label{alg:line:impossible}
      \Comment{outcome impossible under the model}
  \State $\log W^k\gets\log W^k+\log e^k(X_{t+1})$
         for $k=1,\ldots,K$ \label{alg:line:update}
  \State \textbf{if} $\log\widehat W^\LM_{t+1}=\log\frac1K\sum_{k=1}^K\exp(\log W^k)\geq\log(1/\alpha)$ \textbf{then} raise an alarm \label{alg:line:alarm}
\EndFor
\end{algorithmic}
\end{algorithm}

\subsection{Performance Analysis}
\label{sec:regret}
As discussed in \cref{sec:power}, we measure evidence by the log-wealth
$\log W^\theta_t=\sum_{s=1}^t\log e^\theta(X_s;\HH_{s-1})$. If a monitor gains on
average $g>0$ log-wealth per step, i.e., $\log W_t\approx gt$, it crosses
the threshold $\log(1/\alpha)$ after about $\log(1/\alpha)/g$ steps.
We therefore compare the mixture with the monitor for a fixed parameter
$\theta^*$ by the \emph{regret}
$\log W^{\theta^*}_T-\log W^{\LM}_T$, the log-wealth the mixture loses by
not knowing $\theta^*$ in advance. A regret of $R$ means that the mixture
has to accumulate $R$ more log-wealth, which delays its alarm by about
$R/g$ steps. For the Gaussian distribution $\pi$, we show that
the regret is logarithmic in $T$ on every trajectory. The proof follows
Theorem~2.2 of \cite{kakade2004online}, but replaces their curvature
assumption by Lipschitz continuity of the log-payoff in $\theta$.

\begin{restatable}{theorem}{thmMixtureRegret}
\label{thm:mixture-regret}
Assume $\lVert\feat_{t,x}\rVert_\infty\leq B$ for all $t$ and $x$, and let
$(W^\theta_t)_{t\geq0}$ be given by either
\begin{enumerate}
  \item the scale construction $\cC^{\mathrm{Sc}}$, where the robust model is
  any polytopic robust model; or
  \item the RIPr construction $\cC^{\mathrm{RIPr}}$, where the robust model is an
  interval model.
\end{enumerate}
Let $\pi$ be the Gaussian $\mathcal N(0,\sigma^2I_d)$ and let
$W^{\LM}_t$ be the corresponding mixture \eqref{eq:mixture-wealth}.
Then, for every
$\theta^*\in\R^d$ and every $T\geq1$, pathwise,
\begin{equation}
\label{eq:mixture-regret}
  \log W^{\LM}_T
  \geq \log W^{\theta^*}_T
  -\frac{\lVert\theta^*\rVert_2^2}{2\sigma^2}
  -d\log(1+2BT\sigma)-\frac d2.
\end{equation}
\end{restatable}

Since the regret is sublinear in $T$, the mixture monitor assymptotically performs as well as the best fixed-parameter monitor.
Concretely, if the best fixed-parameter monitor accumulates log-wealth at a positive rate $g>0$,
the delay becomes negligible compared to the time $T_\alpha = \log(1/\alpha)/g$
needed to raise the alarm for small $\alpha$.
The following corollary makes this precise.

\begin{restatable}{corollary}{corMixtureDetectionTime}
\label{cor:mixture-detection-time}
Under the assumptions of \cref{thm:mixture-regret},
let $g>0$, let $\theta^*\in\R^d$, and for $\alpha\in(0,1)$, let $\tau_\alpha^{\LM}
  =\inf\{t\geq1:\log W_t^{\LM}\geq\log(1/\alpha)\}$ and $
  T_\alpha = \log(1/\alpha)/g.$

\noindent
Then on every trajectory with
 $ \liminf_{t\to\infty}\frac1t\log W_t^{\theta^*}\geq g$ we have
 $ \limsup_{\alpha\downarrow0}
  \tau_\alpha^{\LM} / T_\alpha
  \leq 1.$
\end{restatable}

\section{Experiments}
\label{sec:experiments}

We evaluate the proposed monitoring framework with a prototype implementation.\footnote{Link hidden for the review process.} We focus on the following questions:
\begin{description}
    \item[\textbf{Q1}] Are the proposed monitors $\alpha$-valid?
    \item[\textbf{Q2}] Does empirical performance support the theoretical advantage of RIPr over Scale?
    \item[\textbf{Q3}] Does the mixture of linear monitors outperform the existing monitors, namely, the transition-wise monitors of~\cite{sethi2026markovian, Fields2025-gu}, and the scoring-function monitors of~\cite{henzinger2025alignment}?
    \item[\textbf{Q4}] How effectively do mixtures of linear monitors exploit observation features?
 \end{description}

\paragraph{Benchmarks}
We use two benchmark suites.
The first contains 18 Markov chains from seven PRISM benchmark families~\cite{Kwiatkowsa2012ThePB}, also used in~\cite{henzinger2025alignment}.
We sample trajectories from each original chain and monitor them against perturbed models generated using nine transition-probability perturbations, similar to~\cite{henzinger2025alignment}.
The second suite contains five variants of 1D and 2D slippery, windy gridworlds: trajectories come from environments with wind, while the models omit it.
Further benchmark details are in Appendix~\ref{sec:appendix-benchmarks}.

We evaluate monitors for both simple and robust models.
A robust model is formed from a perturbed distribution~$p$ using either the multiplicative band $[(1-\varepsilon)p,(1+\varepsilon)p]$ or the additive band $[p-\varepsilon,p+\varepsilon]$, clipped to $[0,1]$, with $\varepsilon\in\{0.05,0.1\}$.
In total, we consider 925 \emph{cases}, each combining a benchmark, a perturbation, and a model variant.

\paragraph{Considered monitors}
(details are described in Appendix~\ref{sec:appendix-monitors}).

\begin{description}
  \item[\texttt{ML}] is the mixture of linear monitors of Section~\ref{sec:algorithm}, instantiated with the features $\feat(h_t)_x = \left(\log(p(x)), p(x), p(x)^2\right)$, where $p = \left(a(h_t) + b(h_t)\right)/2$ for the bounds $a, b$ of the interval model.
  Both constructions are considered, yielding monitors \texttt{ML-Scale} and \texttt{ML-RIPr}.
  The Gaussian distribution $\pi$ is approximated by $300$ parameter vectors drawn from the standard Gaussian distribution.
  We also consider a gridworld-only variant, which uses the features $\feat(h_t)_x = \left(\log(p(x)), c(x)\right)$, where $c(x)$ are the coordinates of the outcome $x$; we consider only its RIPr variant, named \texttt{ML-RIPr-C}.
  \item[\texttt{TW}] is a transition-wise monitor akin to~\cite{Fields2025-gu,sethi2026markovian} that empirically estimates the true transition distribution in each state and computes the RIPr ratio from it.
  We use a robustified version of~\cite{Fields2025-gu}, which  outperforms~\cite{sethi2026markovian}, as shown in Appendix~\ref{sec:appendix-results}.
  \item[\texttt{SR}] is a monitor that uses the scoring rule-based supermartingale of~\cite{henzinger2025alignment}, instantiated with the Brier score and the spherical score.
  We report the spherical score, which performs comparably to or better than the Brier score, as we show in the Appendix~\ref{sec:appendix-results}.
  Since these supermartingales are valid only for simple models, we robustify them using the Scale and the RIPr constructions, yielding \texttt{SR-Scale} and \texttt{SR-RIPr}.
  For simple models, the Scale construction never produces worse bets than those of~\cite{henzinger2025alignment}.
  \item[\texttt{Oracle}] is an idealised reference monitor that has access to the true transition distribution and places asymptotically optimal bets based on it, approximating the theoretical lower bound on the detection time.

\end{description}

\paragraph{Performance metrics}
We measure efficiency by \emph{detection time}, the number of steps until an alarm, and assess validity by checking the empirical \emph{false-alarm rate}.

\paragraph{General setting and hardware}
Following~\cite{henzinger2025alignment}, we sample $300$ trajectories of length $10{,}000$ and run all monitors on the same trajectories.
All monitors use the level $\alpha=0.05$, i.e., they raise an alarm when the log-wealth reaches $\log(1/\alpha)$.
The monitors are implemented in JAX and were run on an NVIDIA GeForce RTX 2060 SUPER GPU.

\paragraph{Monitor latency}
We also ran the monitors on trajectories of $1{,}000$ steps on a single core of an AMD Ryzen 9 3900X CPU.
As predicted in Section~\ref{sec:algorithm}, the per-step computation depends on the number of successors of the current state and importantly does not grow with $\lvert\cX\rvert$ and $\lvert\cA\rvert$.
The slowest environment, \textsf{leader-3-5}, contains only $273$ states, a fraction of our largest benchmark (\textsf{nand-60-4}, 18 M states), yet has up to $125$ successors per state.
For this environment, the average per-step times range from $0.02$\,ms for \texttt{TW} and \texttt{Oracle} to $6.5$\,ms for \texttt{ML-RIPr}, i.e., negligible for monitoring applications.

\subsection*{Q1: Are our monitors $\alpha$-valid?}

\paragraph{Setting and results}
We measure the false-alarm rate of the \texttt{\{SR,ML\}-\{RIPr,Scale\}} and \texttt{TW} monitors on the 90 \texttt{identity} cases, where the model is unperturbed, and the 58 \texttt{model-aligned} cases, in which the perturbed environment is contained in the robust model.
In every case, the empirical false-alarm rate over 300 traces stays below $\alpha$, confirming the theoretical validity of the proposed constructions.
Appendix~\ref{sec:appendix-results} reports the monitor-specific rates, which show no significant differences among the variants.

\subsection*{Q2: The empirical benefits of the RIPr construction.}

\paragraph{Setting}
After removing the \texttt{identity} cases, we retain 835 cases.
Of the 300 traces per case, we pool the 210{,}221 traces in which the \texttt{Oracle} raised an alarm.
We compare each \texttt{RIPr} monitor with its \texttt{Scale} counterpart and with the \texttt{Oracle} on the same trace.
\texttt{TW} and \texttt{ML-RIPr-C} have no \texttt{Scale} counterpart and are not covered.

\paragraph{Results}
\Cref{fig:rq2-ripr-vs-scale} shows that the RIPr construction generally reduces detection times.\footnote{Detection times faster than the \texttt{Oracle} occur occasionally, which is expected: the \texttt{Oracle} is itself a betting monitor and can perform poorly by chance.}
\texttt{SR-RIPr} beats its \texttt{Scale} counterpart on 77.2\% of traces and loses on only 5.0\%, although \texttt{SR-Scale} detects more traces overall.
For \texttt{ML}, the two constructions are barely distinguishable: they often alarm at the same step, and otherwise the gap is small.
This matches the discussion in \cref{sec:bet-geometry}: the rigid parametrisation of the \texttt{SR} monitors cannot place bets whose expectation is maximised on nontrivial faces of the polytopes $\model(h_t)$, in particular the constant bet $e\equiv1$, the only bet that does not lose evidence if $\qstar(\,\cdot\, \mid h_t) \in \model(h_t)$.
The RIPr construction partially mitigates this by transforming a larger set of weights into such bets, whereas the log-linear parametrization of the linear monitors is already flexible enough to approximate them, which also appears to compensate for the inadmissibility of the Scale construction.

\subsection*{Q3: The empirical benefits of the monitors using linear models.}
\paragraph{Setting}
Based on the previous results, we drop the \texttt{Scale} construction and compare \texttt{ML-RIPr} with \texttt{SR-RIPr} and \texttt{TW} on the same pool of 210{,}221 traces as in \textbf{Q2}.

\paragraph{Results}
\Cref{fig:rq3-lm-vs-existing} shows that \texttt{ML-RIPr} detects faster than both competitors: it alarms earlier on 35.3\% of traces against \texttt{SR-RIPr} (later on 11.6\%) and on 20.4\% against \texttt{TW} (later on 12.0\%), with the rest mostly ties.
Moreover, it leaves the fewest traces undetected: 5.0\%, against 15.8\% for \texttt{SR-RIPr} and 21.5\% for \texttt{TW}.
The \texttt{ML} monitors thus retain a clear advantage over the \texttt{SR} monitors even when the latter use the RIPr construction.
\texttt{TW} is competitive only on small models, models whose states are heavily revisited, or when the true transition distributions are close to uniform, which is its base estimate.

\begin{figure}[t]
  \centering
  \begin{subfigure}[t]{0.32\linewidth}
    \centering
    \includegraphics[width=\linewidth]{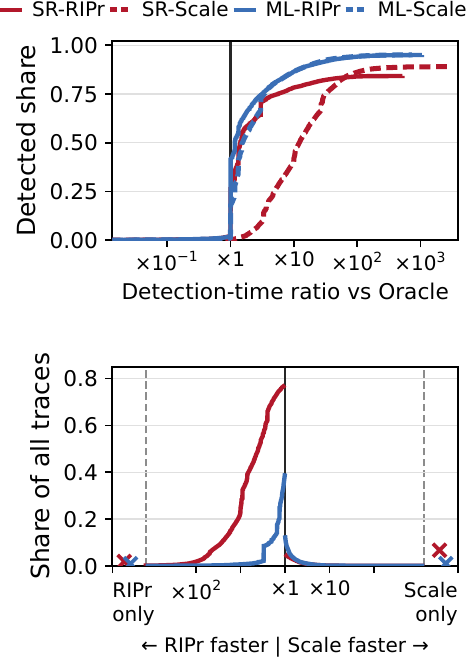}
    \caption{RIPr versus Scale (Q2)}
    \label{fig:rq2-ripr-vs-scale}
  \end{subfigure}\hfill
  \begin{subfigure}[t]{0.32\linewidth}
    \centering
    \includegraphics[width=\linewidth]{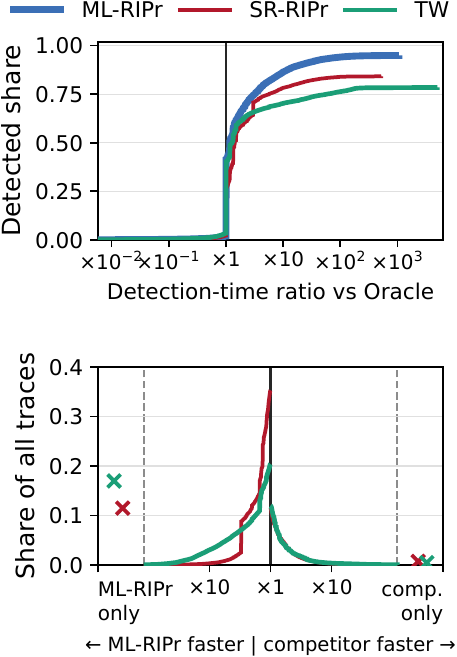}
    \caption{ML versus SR and TW (Q3)}
    \label{fig:rq3-lm-vs-existing}
  \end{subfigure}\hfill
  \begin{subfigure}[t]{0.32\linewidth}
    \centering
    \includegraphics[width=\linewidth]{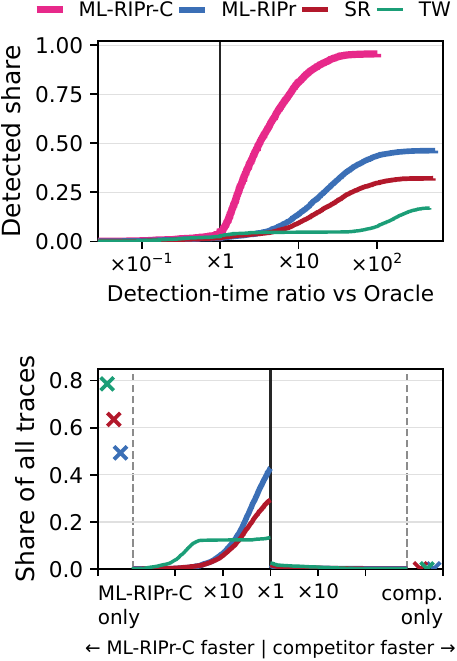}
    \caption{Feature monitoring (Q4)}
    \label{fig:lm-coord}
  \end{subfigure}
  \caption{
All three columns share the same format; only the monitors being compared differ.
  \textbf{Top}: detection time of each monitor relative to the \texttt{Oracle}, i.e., the cumulative share of traces with \(t_{\mathrm{monitor}}/t_{\mathrm{Oracle}}\leq x\), where \(x=1\) marks the \texttt{Oracle}'s time.
  Each curve ends at its monitor's largest finite ratio, at the level corresponding to the share of detected traces.
  \textbf{Bottom}: paired speed-up of the focal monitor against each competitor, on traces detected by both.
  For each threshold \(k\), the curve shows the cumulative share of traces detected more than \(k\)-fold faster: the focal monitor is faster on the left, the competitor on the right.
  Endpoint Xs show the shares of traces detected by only one of the two; traces detected by neither are excluded, and ties are not shown.
  In all panels, 0.1\% of extreme data points are cropped to improve visual clarity.
  The columns compare:
  \textbf{(\subref*{fig:rq2-ripr-vs-scale})} \texttt{\{SR,ML\}-\{RIPr,Scale\}}, pairing each RIPr monitor (focal) with its Scale counterpart;
  \textbf{(\subref*{fig:rq3-lm-vs-existing})} \texttt{ML-RIPr} (focal) against \texttt{SR-RIPr} and \texttt{TW};
  \textbf{(\subref*{fig:lm-coord})} on the gridworld benchmarks with coordinate features, \texttt{ML-RIPr-C} (focal) against \texttt{ML-RIPr}, \texttt{SR-RIPr}, and \texttt{TW}.
  }
  \label{fig:linear-monitor}
\end{figure}

\subsection*{Q4: Can the mixture of linear monitors exploit the available observation features?}

\paragraph{Setting}
We focus on the 25 \texttt{gridworld} cases, which provide coordinate features (the position of the agent in the 1D or 2D grid), and ask whether the coordinate-feature monitor \texttt{ML-RIPr-C} detects faster than the feature-agnostic \texttt{ML-RIPr}, \texttt{SR}, and \texttt{TW}.
Of the 300 traces per case, we pool the 7{,}220 traces in which the \texttt{Oracle} raised an alarm.

\paragraph{Results}
\Cref{fig:lm-coord} demonstrates the superiority of \texttt{ML-RIPr-C} over its competitors.
On traces detected by both monitors, it alarms earlier than \texttt{ML-RIPr}, \texttt{SR}, and \texttt{TW} in 92.4\%, 91.4\%, and 78.8\% of the comparisons, respectively, with a median paired speed-up of at least 2.96$\times$ against each.
Against \texttt{TW}, it is at least an order of magnitude faster on more than 70\% of the shared detections.
Most importantly, feature-awareness enables unprecedented detection coverage: \texttt{ML-RIPr-C} detects nearly twice as many traces as \texttt{ML-RIPr}, 3.0 times as many as \texttt{SR}, and 5.6 times as many as \texttt{TW}, while missing only 0, 4, and 2 traces detected by those monitors, respectively.
\texttt{TW} performs the worst: although it can, in principle, detect any model-environment misalignment, in practice it fails unless the agent repeatedly cycles through a small set of states.

\section{Related Work}
\paragraph{Runtime model monitoring.}  The closest works to ours are the transition-wise monitors of~\cite{Fields2025-gu,sethi2026markovian} and the scoring rule-based monitors of~\cite{henzinger2025alignment}, as these support stochastic systems and anytime-valid guarantees. We explain these methods in detail in Section~\ref{sec:betting-algorithms}. Our bets generalise these two approaches through feature maps that allow sharing evidence across transitions while retaining informative state- and action-dependent features. As our experiments show, this substantially improves detection times over both approaches. More generally, runtime monitoring with stochastic models is a well-established field (see e.g.~\cite{babaee2018predictive,junges2021runtime,cleaveland2023conservative,cairoli2023conformal}), but traditionally, its goal has been monitoring system properties, not checking alignment between a model and environment observations.

\paragraph{Martingales in system analysis.} Test martingales have also been used to detect distribution shifts~\cite{vovk2021retrain}, by accumulating evidence against the exchangeability hypothesis using conformal p-values \cite{vovk2005algorithmic}. However, exchangeability is generally inappropriate for our setting due to the sequential dependencies in the observation sequence.
In contrast, our monitors directly test the model's conditional transition predictions, without assuming exchangeability, stationarity, or Markovianity of the true process.
Supermartingales have been extensively used as certificates for stochastic verification and control, from early work on barrier-based probabilistic safety~\cite{prajna2007framework} to the verification of more general temporal properties~\cite{Chatterjee_Fu_Novotny_2020,cav22martingales,abate2024stochastic,abate2025quantitative,Kaushik26cav,abate26lics}\footnote{We note the contemporaneous work of~\cite{giacobbe2026runtime} (appeared in arXiv only a few days before ours) which develops betting-based monitors similar to ours but for a different problem (certificate monitoring).}. These approaches use supermartingales to establish properties of a stochastic system under assumed dynamics, while our focus is on assessing whether model dynamics remain consistent with runtime observations. Finally, test supermartingales generalise likelihood-ratio-based methods for sequential hypothesis testing~\cite{Wald1945-gv}, widely used in statistical model checking~\cite{younes2006statistical}.

\paragraph{Conformance testing/checking.} These techniques assess the agreement between model and system behaviours, typically through trace relations~\cite{abbas2014formal} or reachable sets~\cite{roehm2016reachset}. While these methods largely focus on non-stochastic or hybrid dynamics~\cite{roehm2019model}, Qin et al.~\cite{qin2023conformance} introduced a stochastic notion of conformance. They estimate this via probabilistic bounds on the distances between two systems' trajectories using offline, data-driven inference, whereas our approach monitors model consistency online, providing anytime-valid guarantees without distributional assumptions beyond model consistency.

\section{Conclusion}

We have described a scalable, universal monitoring framework that can be applied across various contexts such as Markov chains, (partially observable) MDPs, and even multi-agent systems or games. We show theoretical and practical advantages compared to existing monitors. Beyond tabular models, we see the potential to apply our methods to monitoring of learned world models, namely their discrete fragments (as seen, for example, in Dreamer~\cite{hafner2023mastering}) that are inherently equipped with interesting feature vectors.

\subsubsection*{Data availability.}
The experimental results and source code will be made publicly available in a documented repository upon acceptance of the paper.

\bibliographystyle{plainnat}
\bibliography{references}

\newpage
\appendix

\section{Examples}
\label{sec:examples}

\Cref{fig:line-world} restates and visualises the windy line world model from~\cref{ex:line-world}.

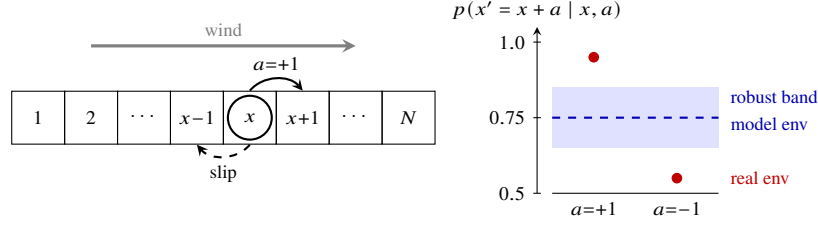
\begin{figure}[t]
  \centering
  \begin{tikzpicture}[
      cell/.style={draw, minimum width=7mm, minimum height=7mm, inner sep=0pt},
      >=stealth]
    \foreach \i/\lab in {0/1, 1/2, 2/\cdots, 3/x{-}1, 4/x, 5/x{+}1, 6/\cdots, 7/N}
      \node[cell] (c\i) at (0.7*\i, 0) {\scriptsize$\lab$};
    \draw[thick] (c4) circle (2.9mm);
    \draw[->, thick] (c4.north) to[bend left=50]
      node[above, font=\scriptsize] {$a{=}{+}1$} (c5.north);
    \draw[->, thick, dashed] (c4.south) to[bend left=50]
      node[below, font=\scriptsize] {slip} (c3.south);
    \draw[->, very thick, gray] (0.7, 0.95) -- (4.2, 0.95)
      node[midway, above, font=\scriptsize, text=gray] {wind};

    \begin{scope}[xshift=6.6cm, yshift=-1.0cm, y=40mm]
      \draw[->] (0, 0) -- (0, 0.55)
        node[above, font=\scriptsize] {$p(x'=x+a\mid x,a)$};
      \foreach \v/\y in {0.5/0, 0.75/0.25, 1.0/0.5}
        \draw (-0.05, \y) -- (0.05, \y) node[left=1mm, font=\scriptsize] {$\v$};
      \fill[blue!12] (0.2, 0.15) rectangle (2.4, 0.35);
      \node[font=\scriptsize, text=blue!70!black, anchor=west] at (2.45, 0.32)
        {robust band};
      \draw[blue!70!black, thick, dashed] (0.2, 0.25) -- (2.4, 0.25);
      \node[font=\scriptsize, text=blue!70!black, anchor=west] at (2.45, 0.23)
        {model env};
      \fill[red!75!black] (0.75, 0.45) circle (2pt);
      \fill[red!75!black] (1.85, 0.05) circle (2pt);
      \node[font=\scriptsize, text=red!75!black, anchor=west] at (2.45, 0.05)
        {real env};
      \draw (0.2, 0) -- (2.4, 0);
      \node[below, font=\scriptsize] at (0.75, 0) {$a{=}{+}1$};
      \node[below, font=\scriptsize] at (1.85, 0) {$a{=}{-}1$};
    \end{scope}
  \end{tikzpicture}
  \caption{\emph{Windy line world}. Left: the states are the cells
  $\cX=\{1,\dots,N\}$ and the actions are $\cA=\{-1,+1\}$ (left/right).
  From cell $x$, the agent reaches the intended cell $x+a$ or
  slips to $x-a$; a move leaving the line keeps the agent in place.
  Right: the probability $p(x+a\mid x,a)$ that the intended move succeeds. In
  the model environment (dashed), it is $0.75$ for both actions. The robust
  model $\model$ is an interval MDP that widens every transition
  probability $p$ to $[p-0.1,p+0.1]$ (shaded). In the real environment (dots), the wind shifts
  probability $0.2$ to the right neighbour, so the intended move succeeds with
  probability $0.95$ for $a=+1$ and $0.55$ for $a=-1$. Both probabilities lie outside the
  band, so the true process is not consistent with $\model$;
  in both cases the model underestimates the
  probability of moving right.}
  \label{fig:line-world}
\end{figure}

\section{Test Supermartingales and Betting Monitors}
\label{app:test-sm}
In this section, we introduce test supermartingales, which play an important
role in our monitors. They were first studied by Ville~\cite{ville1939etude}
in the context of sequential hypothesis testing although named only later by \citet{Shafer2011-xg}.
We present its robust version as formulated in \cite{Ramdas2025-gq}.

We begin by introducing necessary notation.
Every process $(X_t,A_t)_{t\geq0}$ consistent with $\model$ induces a distribution $\Prob$ on the set
$(\cX\times\cA)^{\N}$ of trajectories; we call such a distribution \emph{consistent} with $\model$ and write
$\cP_\model$, or simply $\cP$, for the set of consistent distributions.
Since $\cX$ and $\cA$ are finite, a history is consistent with $\model$ if and only if it has
positive probability under some $\Prob\in\cP_\model$, so ``on every history consistent with $\model$''
means ``$\Prob$-almost surely for every $\Prob\in\cP_\model$''.

The monitor observes $X_t$ before $A_t$ is chosen, so we use the \emph{natural filtration}
$\cF_t=\sigma(\HH_{t-1},X_t)$. A process is adapted in the sense of \cref{sec:setup} if and only if
it is adapted to $(\cF_t)_{t\geq0}$.
\begin{definition}
\label{def:test-supermartingale}
A \emph{test supermartingale for $\model$} is a non-negative adapted process
$(W_t)_{t\geq0}$ with $W_0=1$ such that, for every $t\geq0$ and every $\Prob\in\cP_\model$,
\begin{equation}
  \label{eq:additive-supermartingale}
  \E_\Prob[W_{t+1}\mid\cF_t]\leq W_t
  \qquad\text{$\Prob$-almost surely.}
\end{equation}
\end{definition}
We call \eqref{eq:additive-supermartingale} the \emph{supermartingale property}.
Whenever $W_t>0$, we can divide by $W_t$, which is $\cF_t$-measurable, and obtain the
\emph{multiplicative supermartingale property}
\begin{equation}
  \label{eq:supermartingale}
  \E_\Prob\Bigl[\frac{W_{t+1}}{W_t}\Bigm|\cF_t\Bigr]\leq 1.
\end{equation}
Whenever $W_t=0$, the supermartingale property and non-negativity force
$W_{t+1}=0$ $\Prob$-almost surely, so the process remains at zero.

\subsection{From Test Supermartingales to Betting Monitors}
The multiplicative supermartingale property \eqref{eq:supermartingale} has a
natural interpretation as a betting game between the
model and a bettor. The model believes that the dynamics are consistent with
$\model$ and accepts any bet that is fair under every distribution it
considers plausible. The bettor bets against the model and aims to
accumulate as much wealth as possible. At each step $t$, the bettor stakes
their entire wealth $W_t$ and announces a payoff function
$e\colon\cX\to[0,\infty]$. If the next state is
$X_{t+1}=x$, the bettor receives $W_t \cdot e(x)$, so that the new
wealth is $W_{t+1}=W_t \cdot e(X_{t+1})$. The model accepts the bet only
if it cannot lose money in expectation under any distribution plausible according to $\cU = \model(h_t)$,
i.e., if
\begin{equation}
\label{eq:fair-bet-restated}
  \E_{x\sim p}[e(x)]\leq1
  \qquad\text{for every }p\in\cU.
\end{equation}
We call any such function $e\colon\cX\to[0,\infty]$ satisfying \eqref{eq:fair-bet-restated}
a \emph{valid} or \emph{fair} bet for $\cU$.%
With the initial budget $W_0=1$, and repeating this procedure with a fair bet
$e(\,\cdot\,;h_{i-1})$ for $\model(h_{i-1})$ in every step $i$,
the bettor's wealth after $t$ steps is
\begin{equation}
\label{eq:product-supermartingale}
  W_t = \prod_{i=1}^{t} e(X_{i};\HH_{i-1}),
\end{equation}
where each bet is chosen based on the
history $h_{i-1}$. Since every bet is fair from the model's perspective,
the wealth process $(W_t)_{t\geq0}$ is a test supermartingale. Conversely, every test supermartingale
is the wealth of a bettor by choosing the fair bet $e(X_{t+1};H_t)=W_{t+1}/W_t$.
This yields an important characterization of test supermartingales:

\begin{lemma}
\label{thm:betting-decomposition}
  A process $(W_t)_{t\geq0}$ is a test supermartingale for $\model$
  if and only if it can be decomposed as in \eqref{eq:product-supermartingale} for a bettor $e(\,\cdot\,;h_t)$
  that is fair for every history consistent with $\model$.
\end{lemma}
\begin{proof}
  Proposition~3 of \citet{Waudby-Smith2024-dt} states that $(W_t)_{t\geq0}$ is a test supermartingale for $\model$ if and only if
  $W_t=\prod_{i=1}^tY_i$ for non-negative adapted $Y_i$ such that $\E_\Prob[Y_{t+1}\mid\cF_t]\leq1$
  $\Prob$-almost surely for every $\Prob\in\cP_\model$.

  Therefore, any wealth process $(W_t)_{t\geq0}$ of a bettor is a test supermartingale,
  since by letting $Y_{t+1}=e(X_{t+1};\HH_t)$, we have $\E_\Prob[Y_{t+1}\mid\cF_t] = \E_\Prob[\E_\Prob[e(X_{t+1};\HH_t)\mid\HH_t]\mid\cF_t] \leq \E_\Prob[1\mid\cF_t] = 1$ for every $\Prob\in\cP_\model$.
  The first equality follows from the tower property of conditional expectation, and the inequality from the fairness of the bet $e(\,\cdot\,;\HH_t)$ for $\model(\HH_t)$.

  Conversely, if $(W_t)_{t\geq0}$ is a test supermartingale, let $(Y_t)_{t\geq1}$ be as above.
  Since $Y_{t+1}$ is adapted, it is a function of $(\HH_t,X_{t+1})$, and we define the bettor's bet by $e(x;h_t)=Y_{t+1}(h_t,x)$,
  so that $W_t=\prod_{i=1}^te(X_i;\HH_{i-1})$.
  We need to show that the bets are fair for every history consistent with $\model$.
  Let $h_{t}$ be a history consistent with $\model$ that decomposes into $h_{t-1}$ followed by $x_{t}$ and $a_{t}$,
  and let $p \in \model(h_t)$.
  Since the actions are unconstrained and $\model$ constrains each transition separately, there is
  $\Prob\in\cP_\model$ such that $\Prob(\HH_t=h_t)>0$, $A_t=a_t$ almost surely given $(\HH_{t-1},X_t)=(h_{t-1},x_t)$,
  and $\Prob(X_{t+1}=\cdot\mid\HH_t=h_t)=p$.
  Then, we have
  \begin{align*}
    1 &\geq \E_\Prob[Y_{t+1}\mid\HH_{t-1}=h_{t-1}, X_t=x_t]\\
    &= \E_\Prob[Y_{t+1}\mid\HH_{t-1}=h_{t-1}, X_t=x_t, A_t=a_t]\\
    &= \E_\Prob[e(X_{t+1};h_t)\mid\HH_t=h_t]\\
    &= \E_{x\sim p}[e(x;h_t)],
  \end{align*}
  where the inequality is the property of $Y_{t+1}$ above,
  the first equality follows from $A_t=a_t$ almost surely given $(\HH_{t-1},X_t)=(h_{t-1},x_t)$ under $\Prob$,
  the second equality is the definition of $e(\,\cdot\,;h_t)$, and the last equality follows from the choice of $\Prob$.
  Since $p\in\model(h_t)$ was arbitrary, we conclude that $e(\,\cdot\,;h_t)$ is a fair bet for $\model(h_t)$.
  \qed
\end{proof}

Thus, we can indeed identify test supermartingales with betting monitors. It remains to determine the wealth threshold at which such a monitor should raise an alarm in order to stay $\alpha$-valid.

\subsection{Anytime Validity}
The intuitive correspondence between wealth and evidence can be made
precise. Consider any process consistent with $\model$. The
supermartingale property gives $\E[W_{t+1}\mid\cF_t]\leq W_t$ almost
surely; in particular, $\E[W_{t+1}]\leq\E[W_t]$. By induction,
\(
  \E[W_{t+1}]\leq \E[W_0]=1
\)
for every $t\geq0$.
Since $W_t$ is non-negative, Markov's inequality applies and gives
$\Prob\!\left(W_t\geq1/\alpha\right)\leq\alpha\,\E[W_t]\leq\alpha$ for every fixed $t\geq0$.
A monitor, however, runs for an indefinite amount of time and
inspects the wealth after every step, so a bound for each fixed time does
not suffice; {for instance, a union bound would be vacuous, while a Bonferroni correction would require increasingly conservative thresholds. We ask: under a process consistent with $\delta$, given enough time, could the wealth eventually exceed
$1/\alpha$ with probability greater than $\alpha$? A well-known result
due to Ville~\cite{ville1939etude} shows that this cannot happen; the same bound holds uniformly over all times.}

\begin{proposition}[Ville's inequality]
\label{prop:ville}
If $(W_t)_{t\geq0}$ is a test supermartingale for $\model$, then, for every
$\alpha\in(0,1)$ and every process consistent with $\model$ it holds that
$\Prob\left(\sup_{t\geq0}W_t\geq\frac1\alpha\right)\leq\alpha$.
Consequently, the process $(\monitor_t)_{t\geq0}$ with $\monitor_t=\ind{\sup_{s\leq t}W_s\geq1/\alpha}$ is an
$\alpha$-valid monitor.
\end{proposition}

{
Beyond validity, test supermartingales have further practical advantages.
For example, any construction of a test supermartingale yields, essentially for free, a so-called
e-detector \cite{shin2023detectors}, which detects a change point when the true process is
initially consistent with the model but departs from it at an unknown time.
Moreover, unlike classical hypothesis tests which can suffer inflated false-positive rates when $p$-values are repeatedly inspected (so-called \textit{peeking}), test supermartingales allow history-dependent decisions about when to stop collecting observations, without compromising validity. We refer the reader to \cite{Ramdas2025-gq} for a detailed discussion of advantages of supermartingales.

Finally, test supermartingales generalise likelihood-ratio-based methods for sequential hypothesis testing, notably the sequential probability ratio test~\cite{Wald1945-gv}, widely used in statistical model checking~\cite{younes2006statistical}. Indeed, the cumulative likelihood ratio between an alternative and a null hypothesis is non-negative, and under the null its expected multiplicative increase is bounded by $1$, so it is a test supermartingale.
}
\subsection{Completeness of Test Supermartingales and Betting Monitors}

The betting interpretation suggests a simple way to accumulate evidence: a
test supermartingale multiplies one-step bets, each designed to extract as much
evidence as possible given the current knowledge. It is not clear a priori,
however, that this way of aggregating evidence is optimal, i.e., that no
other valid monitor can combine the observations more cleverly and raise an
alarm faster than every test supermartingale. Indeed, in general sequential
hypothesis testing, there are settings in which test supermartingales are not
optimal and the appropriate tool is a more general class of processes called e-processes
\cite{Ramdas2022-jn,ramdas2020admissible,Grunwald2024-pk}. In~\Cref{thm:monitor-supermartingale} below we show that for robust models considered here, this does not happen: test supermartingales form a complete language for monitoring, in the sense that
every valid monitor is dominated by a monitor based on a test supermartingale.

\begin{restatable}{theorem}{thmMonitorSupermartingale}
\label{thm:monitor-supermartingale}
Let $(\monitor_t)_{t\geq0}$ be an $\alpha$-valid monitor for a polytopic robust model $\model$.
Then there is a test supermartingale $(W_t)_{t\geq0}$
for $\model$ that dominates the rescaled alarm process under $\model$, i.e.,
$W_t\geq\frac{\monitor_t}{\alpha}$ for every $h_t$ consistent with $\model$. In particular,
the monitor $(\ind{W_t\geq1/\alpha})_{t\geq0}$ raises its first alarm no
later than $(\monitor_t)_{t\geq0}$ on every history consistent with $\model$.
\end{restatable}

We prove the above theorem using the machinery of sequential testing in Appendix~\ref{app:proofs-monitoring}. Given the previously established equivalence between test supermartingales and betting monitors, \Cref{thm:betting-monitor} follows immediately.

\section{Learning Interpretation of Mixture Monitors}
\label{app:mixtures}
The mixture wealth process $W^\LM$ of \eqref{eq:mixture-wealth} evolves as $W^\LM_{t+1}=W^\LM_t\,e^\LM(X_{t+1};\HH_t)$ with the bets
\[
  e^\LM(\,\cdot\,;h_t)=\int_{\R^d}\pi_t(\theta)\,e^\theta(\,\cdot\,;\HH_t)\,\mathrm{d}\theta,
  \qquad
  \pi_t(\theta)\propto\pi(\theta)\,W^\theta_t.
\]
Since every such bet is a convex mixture of fair bets it is itself a fair bet due to Tonelli's theorem. Note that the weight of $\pi_t$ shifts towards the parameters that have performed best so far, so the mixture method can be seen as a learning technique that uses a distributional instead of a point estimate.

\section{Details of the Considered Monitors}
\label{sec:appendix-monitors}

\paragraph{Transition-wise monitor (\texttt{TW})}
The sequential test of \citet{Fields2025-gu} bets sequentially. In every
step, it uses the empirical estimate of the transition distribution of the
current state to place an assymptotically optimal bet.
However, it does not consider robust models.
\citet{sethi2026markovian} do consider robust models, but while
their monitor is asymptotically optimal, it is too slow in practice.
We therefore robustify the approach of \citet{Fields2025-gu} and apply
the RIPr construction to the estimate of the true transition distribution.

Formally, for a state $u$, let $N_t(u,y)$ denote the number of transitions
from $u$ to $y$ observed in $h_t$, $N_t(u)=\sum_y N_t(u,y)$, and
$d_u=|\supp(\model(h_t))|$ the number of possible successors of $u$. When the
current state is $u=x_t$, the monitor estimates the true transition
distribution by the Laplace estimate
\[
  \widehat q_{t+1}(y)=\frac{N_t(u,y)+1}{N_t(u)+d_u},
  \qquad y\in\supp(\model(h_t)),
\]
and bets with the RIPr ratio
$e(x;h_t)=\widehat q_{t+1}(x)/p^*(x)$, where
$p^*\in\RIPr(\widehat q_{t+1},\model(h_t))$. Since $\widehat q_{t+1}$ is
determined by $h_t$, this is a fair bet for $\model(h_t)$ by
\Cref{thm:ripr-well-defined}. Moreover, as $\widehat q_{t+1}$ converges to the true transition distribution,
these bets become asymptotically optimal.

\paragraph{Scoring rule monitor (\texttt{SR})} uses fair bets introduced in \cite{henzinger2025alignment} that are based on scoring rules;
functions $f(x, p): \cX \times \dist{\cX} \to [a, b] \subseteq \R$ that assign a numerical score of how well
a prediction is aligned with the observed outcome. The paper uses two scoring rules: the Brier score (\texttt{SR-Br}) and the spherical score (\texttt{SR-Sph}).
For the given scoring function, the bet is then defined as
$$e(X) = \exp\left(\lambda(F-\mu) - \kappa(F - \gamma)^2\right)$$
where $X$ is the random outcome, $F = f(X, p)$, $\mu = \E_{X \sim p}[f(X, p)]$,
$\gamma$ is the running mean of the past scores,
and $\kappa$ is a suitable constant.
The values of $\gamma$ and $\lambda$ are thus parameters.
For $\gamma$, we use the running mean of the past scores as in \cite{henzinger2025alignment};
however, the value of $\lambda$ is not discussed in the paper,
hence we consider the Bayesian mixture over several such values as in the mixture of linear monitors.

The bet $e(X)$ is a fair bet only for a simple model, since it depends on the true expected score $\mu$
which is unknown in the robust model setting. To address this,
we define weights $w_x = E(x)$ and apply the \emph{Scale} constructions
described in Sec.~\ref{sec:scale} and the \emph{RIPr} construction described in Sec.\ref{sec:ripr}
to derive a fair bet $e'(X)$ for the robust model. Clearly, the scale construction can only improve a fair bet,
hence for simple models it holds that $e'(X) \geq e(X)$.

\paragraph{Oracle monitor (\texttt{Oracle})}
The \texttt{Oracle} knows the true transition distribution
$\qstar_{t+1}(\,\cdot\mid h_t)$. Following \cref{sec:power}, in every step
it places the fair bet for $\model(h_t)$ with the largest growth rate
$\E_{\qstar_{t+1}}[\log e]$. This bet is exactly the RIPr ratio
\cite{Grunwald2024-pk,Lardy2024-gy}
\[
  e(x;h_t)=\frac{\qstar_{t+1}(x\mid h_t)}{p^*(x)},
  \qquad
  p^*\in\RIPr\bigl(\qstar_{t+1}(\,\cdot\mid h_t),\model(h_t)\bigr),
\]
i.e., the RIPr construction applied to the weights
$w=\qstar_{t+1}(\,\cdot\mid h_t)$.

\section{Details of the Benchmarks}
\label{sec:appendix-benchmarks}

All benchmarks are discrete-time Markov chains given as PRISM programs and
built with the probabilistic model checker Storm~\cite{hensel2022storm}.
\Cref{tab:benchmarks} lists their sizes.

\begin{table}[t]
  \centering
  \caption{Sizes of the benchmarks: the number of states, the number of
    transitions, the maximal number of successors of a state, and the number
    of terminal states.}
  \label{tab:benchmarks}
  \small
  \begin{tabular}{llrrrr}
    \toprule
    & benchmark & states & transitions & max.\ succ. & terminal \\
    \midrule
    \multirow{8}{*}{\rotatebox{90}{Storm}}
    & \textsf{brp-16-2}    & 677    & 867    & 2   & 35   \\
    & \textsf{conditional} & 4      & 5      & 2   & 2    \\
    & \textsf{crowds-4-3}  & 726    & 1\,146 & 4   & 35   \\
    & \textsf{crowds-5-5}  & 8\,607 & 15\,113 & 5  & 1\,260 \\
    & \textsf{die}         & 13     & 20     & 2   & 6    \\
    & \textsf{leader-3-5}  & 273    & 397    & 125 & 1    \\
    & \textsf{nand-5-2}    & 1\,728 & 2\,505 & 2   & 6    \\
    & \textsf{quantiles}   & 4      & 6      & 3   & 1    \\
    \midrule
    \multirow{10}{*}{\rotatebox{90}{scaled Storm}}
    & \textsf{brp-32-3}    & 1\,766      & 2\,307       & 2  & 68 \\
    & \textsf{brp-64-5}    & 5\,192      & 6\,915       & 2  & 134 \\
    & \textsf{brp-256-5}   & 20\,744     & 27\,651      & 2  & 518 \\
    & \textsf{brp-1024-5}  & 82\,952     & 110\,595     & 2  & 2\,054 \\
    & \textsf{crowds-10-5} & 111\,294    & 261\,444     & 10 & 3\,003 \\
    & \textsf{crowds-20-5} & 2\,061\,951 & 7\,374\,951  & 20 & 53\,130 \\
    & \textsf{crowds-20-6} & 10\,633\,591 & 38\,261\,191 & 20 & 230\,230 \\
    & \textsf{nand-20-2}   & 154\,942    & 239\,832     & 2  & 21 \\
    & \textsf{nand-40-4}   & 3\,999\,522 & 6\,288\,542  & 2  & 41 \\
    & \textsf{nand-60-4}   & 18\,826\,082 & 29\,772\,212 & 2 & 61 \\
    \midrule
    \multirow{5}{*}{\rotatebox{90}{grid}}
    & \textsf{windy}            & 3\,196  & 12\,760  & 4  & 4 \\
    & \textsf{windy-obs}        & 3\,156  & 12\,600  & 4  & 4 \\
    & \textsf{windy-obs-moving} & 29\,504 & 465\,320 & 36 & 704 \\
    & \textsf{line}             & 398     & 794      & 2  & 2 \\
    & \textsf{line-large}       & 4\,998  & 9\,994   & 2  & 2 \\
    \bottomrule
  \end{tabular}
\end{table}

\paragraph{Storm benchmarks}
The first suite consists of the eight Markov chains from the Storm benchmark
suite\footnote{\url{https://github.com/moves-rwth/storm/tree/master/resources/examples/testfiles/dtmc}}
that were considered by \citet{henzinger2025alignment}: the bounded
retransmission protocol \textsf{brp}, the \textsf{crowds} anonymity protocol,
Knuth's \textsf{die}, the synchronous \textsf{leader} election protocol, the
NAND multiplexing \textsf{nand}, and two small test models
\textsf{conditional} and \textsf{quantiles}. To evaluate scalability, we add
ten larger instances of \textsf{brp}, \textsf{crowds}, and \textsf{nand},
obtained from the same PRISM programs by increasing their parameters.

\paragraph{Perturbations}
Trajectories are sampled from the original chain with transition matrix $P$,
and the model is obtained by perturbing $P$ row by row. We use the nine
perturbations of \citet{henzinger2025alignment}. After each perturbation,
the values are clipped from below at $10^{-3}$ and renormalised to form a probability distribution.
\begin{description}
  \item[\textsf{sharp}] raises every entry of a row to the power $4$.
  \item[\textsf{supp}] randomly redistributes the probability mass of a row.
    Every transition keeps its original probability with probability $0.3$
    (and at least one randomly chosen transition always keeps it); every
    other transition gets a new random weight drawn uniformly from $[0,1]$.
  \item[\textsf{noisy}] adds $R$ with $R\sim\mathrm{Uniform}([-0.05,0.05])$
    to every entry.
  \item[\textsf{drop}] sets every entry to $10^{-3}$ with probability $0.4$,
    keeping at least one randomly chosen entry of the row.
  \item[\textsf{swap}] swaps the largest and the smallest entry of a row;
  \item[\textsf{collapse}] puts nearly all mass of a row on one randomly chosen
    entry: the chosen entry gets weight $1$, all others $10^{-3}$.
  \item[\textsf{bias}] mixes a row with the point mass on its first successor
    (in the order of the state indices), with weight $0.55$ on the point mass.
  \item[\textsf{inv}] replaces every entry $p$ by $1/p$ and renormalises, so
    that low probabilities become dominant.
  \item[\textsf{flip}] replaces every entry $p$ by $1-p$.
\end{description}
In addition, we include the unperturbed chain (\textsf{identity}) as the model
to evaluate false alarms. The random perturbations are drawn once per
benchmark with a fixed seed.
Our perturbations differ from those of \citet{henzinger2025alignment} in that
they preserve the support of every row. This is consistent with our assumption
from \Cref{sec:robust-model} that $\supp(\qstar(h_t))\subseteq\supp(\model(h_t))$.

\paragraph{Terminal states}
To avoid bottom strongly connected components, \citet{henzinger2025alignment}
add a transition to every state with probability $0.01$ to the initial state.
We do not modify the chains in this way, since it changes the dynamics of both
the system and the model. Instead, whenever a trajectory enters a terminal
state, i.e., a state whose only transition is a self-loop (Storm represents
deadlocks in the same way), the next transition starts from the initial state
again. The transition into the terminal state is part of the trajectory, while
the reset is not observed by the monitor.

\paragraph{Gridworlds}
The second suite consists of five grid worlds in which an agent moves under a
uniformly random policy. Every move takes two steps of the chain: first, the
agent chooses one of the four directions (up, down, left, right) uniformly at
random, and then it moves. With probability $0.75$, the agent moves in the
chosen direction, and with probability $0.25$, it slips to one of the three
other directions, chosen uniformly at random; a move that would leave the grid
or enter an obstacle keeps the agent in place. In the system, a wind blows the
agent down: when the chosen direction is up, left, or right, the agent is moved
down with an additional probability $0.2$, taken from the intended move; when
the chosen direction is down, the probability of moving down increases by $0.2$,
taken from the slips. The model is the same grid world without wind. Since the
wind is weaker than the slips, the model and the system have the same support.
\begin{description}
  \item[\textsf{windy}] is a $20\times 20$ grid; the agent starts in the
    bottom-left corner, and the trajectory is reset when it reaches the
    top-right corner.
  \item[\textsf{windy-obs}] adds five static obstacles, which the agent cannot
    enter, to the \textsf{windy} grid.
  \item[\textsf{windy-obs-moving}] is an $8\times 8$ grid with five static
    obstacles and two moving ones; one moves vertically and the other
    horizontally, by one cell in each direction with probability $0.2$ per move
    of the agent. The trajectory is reset when the agent reaches the
    top-right corner or collides with a moving obstacle.
  \item[\textsf{line}] is the one-dimensional variant of \textsf{windy}
    shown in \cref{ex:line-world}: a line of $100$ cells, where the agent
    chooses between left and right and the wind blows to the right.
  \item[\textsf{line-large}] is the same line world with $1250$ cells.
\end{description}

\section{Further Experimental Results}
\label{sec:appendix-results}

\Cref{fig:validity-and-projection} reports cumulative false alarm rates for the experiment in \cref{sec:experiments} (\textbf{Q1}); \Cref{fig:appendix-hist} extends the cumulative histograms of \Cref{fig:rq2-ripr-vs-scale,fig:rq3-lm-vs-existing,fig:lm-coord} (\textbf{Q2-Q4}).
Both figures also include results for the \texttt{SR} monitors that use the Brier scoring rule.

\begin{figure}[t]
  \centering
 \begin{subfigure}[t]{0.48\linewidth}
   \centering
   \includegraphics[width=\linewidth]{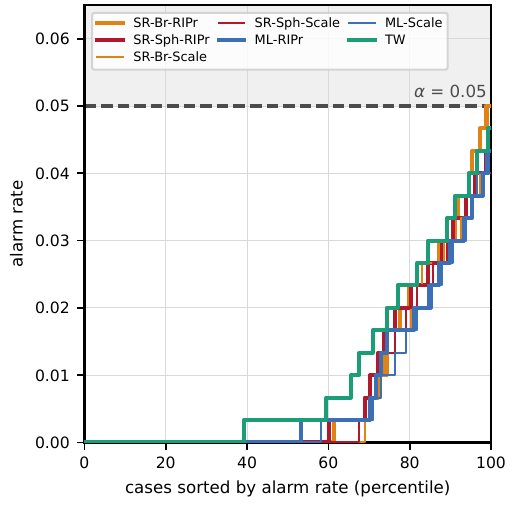}
   \label{fig:rq1-false-alarm-rate}
 \end{subfigure}
  \caption{
  Cumulative false-alarm rate. Each line shows case alarm rates, ordered by rate and plotted as case percentiles. Lower is better.
}
  \label{fig:validity-and-projection}
\end{figure}

\begin{figure}[!htbp]
  \centering
  \begin{subfigure}[t]{\linewidth}
    \centering
    \includegraphics[width=\linewidth]{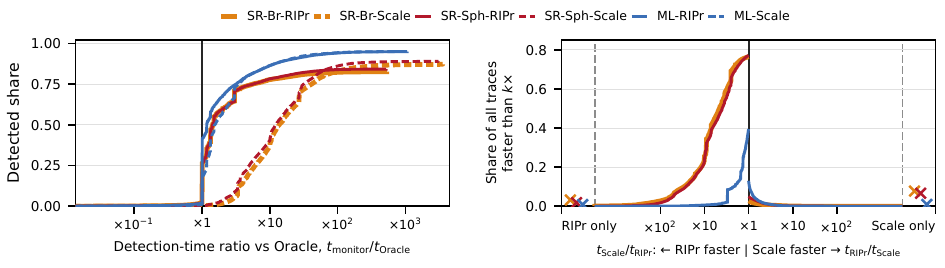}
    \caption{RIPr versus Scale (as \Cref{fig:rq2-ripr-vs-scale})}
    \label{fig:rq2-app}
  \end{subfigure}

  \medskip
  \begin{subfigure}[t]{\linewidth}
    \centering
    \includegraphics[width=\linewidth]{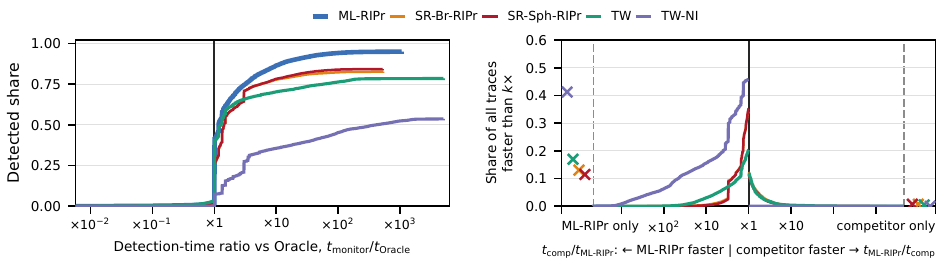}
    \caption{ML versus SR and TW (as \Cref{fig:rq3-lm-vs-existing})}
    \label{fig:rq3-app}
  \end{subfigure}

  \medskip
  \begin{subfigure}[t]{\linewidth}
    \centering
    \includegraphics[width=\linewidth]{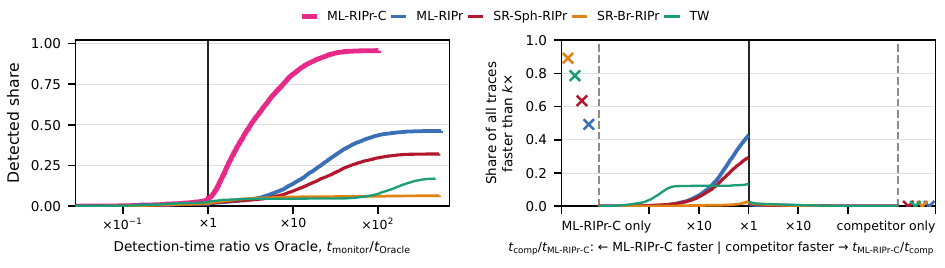}
    \caption{Coordinate features on the gridworlds (as \Cref{fig:lm-coord})}
    \label{fig:rq4-app}
  \end{subfigure}
  \caption{Extension of the cumulative histograms of \Cref{fig:rq2-ripr-vs-scale,fig:rq3-lm-vs-existing,fig:lm-coord}
with the Brier scoring rule of the \texttt{SR} monitors.
\texttt{SR-Br-RIPr} behaves similarly to \texttt{SR-Sph-RIPr} on all benchmarks, except on the gridworlds (\subref*{fig:rq4-app}), where \texttt{SR-Sph-RIPr} never detects $66\%$ of the traces against $93\%$ for \texttt{SR-Br-RIPr}. \texttt{SR-Br-RIPr} is also significantly slower on traces detected by both variants. (\subref*{fig:rq3-app}) also includes comparison to the non-incremental transition-wise monitor of~\cite{sethi2026markovian}, labelled as \texttt{TW-NI}; it is excluded from (\subref*{fig:rq4-app}) where it detects zero traces.
  \label{fig:appendix-hist}}
\end{figure}

\newpage
\section{Proofs}
\label{app:proofs}
This appendix collects the proofs of the results stated in the main text.

\subsection{Additional Notation}
For brevity, we use the notation $\clip(v, a, b) := \min\left\{\max\left\{v, a\right\}, b\right\}$ to denote the clipping of a value $v$ to the interval $[a, b]$.

Whenever $(W_t)_{t\geq0}$ and $(V_t)_{t\geq0}$ are two adapted processes,
we say that $(W_t)_{t\geq0}$ \emph{dominates} $(V_t)_{t\geq0}$ under $\model$
if $V_t\leq W_t$ on every history consistent with $\model$.

We also need a generalization of test supermartingales, called e-processes.
An \emph{e-process} for $\model$ is a non-negative adapted process
$(E_t)_{t\geq0}$ such that, for every $\Prob\in\cP_\model$, there is a
test supermartingale $(V_t)_{t\geq0}$ for $\Prob$ that dominates $(E_t)_{t\geq0}$
$\Prob$-almost surely \cite[Def.~7.3]{Ramdas2025-gq}.
In general, $(V_t)_{t\geq0}$ depends on $\Prob$. \citet{Ramdas2022-jn}
call such processes $\cP$-safe e-processes.

A set $\mathcal{G}$ of distributions on trajectories is \emph{locally dominated} by a distribution $R$ if, for
every $t$, the law of $\HH_t$ under every $\Prob\in\mathcal{G}$ is absolutely
continuous with respect to its law under $R$.
In the case of finite $\cX$ and $\cA$ this means that every history with positive probability under $\Prob \in \mathcal{G}$
has also positive probability under $R$. This holds trivially for every set of distributions,
including $\cP$, since there exists a distribution that assigns a positive probability to every history $h_t$.

Finally, we introduce the concept of \emph{fork-convexity} \cite[Def.~10]{Ramdas2022-jn}, phrased in terms of
the observations. For $s\geq0$, let $g_s=(h_{s-1},x_s)$ and $g_{-1}$ be the empty sequence,
and write $\Prob(\,\cdot\mid g_s)$ for conditioning on the event
$(\HH_{s-1},X_s)=g_s$. Further, let $\Prob,\Prob'$ be distributions on trajectories, $s\geq-1$,
and let $\lambda$ assign to every $g_s$ a weight $\lambda(g_s)\in[0,1]$ with $\lambda(g_s)=1$
whenever $\Prob'(g_s)=0$.
The \emph{fork-convex combination} of $\Prob$ and $\Prob'$ at time $s$ with weights
$\lambda$ is the distribution $\Prob''$ that follows $\Prob$ up to $X_s$ and then
mixes the continuations of $\Prob$ and $\Prob'$ with weight $\lambda(g_s)$, i.e.,
$\Prob''(g_s)=\Prob(g_s)$ and, for every $t\geq s$ and every history $h_t$
extending $g_s$ with $\Prob(g_s)>0$,
\[
  \Prob''(\HH_t=h_t\mid g_s)
  =\lambda(g_s)\,\Prob(\HH_t=h_t\mid g_s)
  +\bigl(1-\lambda(g_s)\bigr)\,\Prob'(\HH_t=h_t\mid g_s).
\]
A set $\cG$ of distributions on trajectories is \emph{fork-convex} if it contains every
fork-convex combination of any two of its elements.
Note that, in particular, $\cP_\model$ is fork-convex whenever every $\model(h_t)$ is
convex. Indeed, the fork-convex combination may choose the actions arbitrarily, and
for $t\geq s$ and every $h_t$ with $\Prob''(\HH_t=h_t)>0$, its
conditional transition distribution satisfies
\[
  \Prob''(X_{t+1}=\cdot\mid \HH_t=h_t)
  =\mu\,\Prob(X_{t+1}=\cdot\mid \HH_t=h_t)+(1-\mu)\,\Prob'(X_{t+1}=\cdot\mid \HH_t=h_t),
\]
where
\[
  \mu=\frac{\lambda(g_s)\,\Prob(\HH_t=h_t\mid g_s)}{\Prob''(\HH_t=h_t\mid g_s)},
\]
and is therefore a member of $\model(h_t)$ whenever $\Prob,\Prob'\in\cP_\model$ and
$\model(h_t)$ is convex; for $t<s$, it coincides with that of $\Prob$.

\subsection{Proof of~\Cref{thm:betting-validity}}
\begin{proof}[Proof of \cref{thm:betting-validity}]
  The wealth of a betting monitor is a product \eqref{eq:product-supermartingale} of bets
  that are fair for every history consistent with $\model$, so it is a test
  supermartingale for $\model$ by \cref{thm:betting-decomposition}, and the monitor is
  $\alpha$-valid by \cref{prop:ville}. \qed
\end{proof}

\subsection{Proof of~\Cref{thm:monitor-supermartingale} (and thus, of~\Cref{thm:betting-monitor})}
\label{app:proofs-monitoring}

\thmMonitorSupermartingale*
\begin{proof}[Proof of \cref{thm:monitor-supermartingale}]
  By our definition, an $\alpha$-valid monitor $(\monitor_t)_{t\geq0}$ is precisely
  a level-$\alpha$ sequential test for $\cP$ as in eq.~(7.1) of \cite{Ramdas2025-gq}.
  Proposition~7.8 \cite{Ramdas2025-gq} states that every level-$\alpha$ sequential test for $\cP$ is in fact a thresholded e-process for $\cP$,
  i.e., $E_t:=\monitor_t/\alpha$ is an e-process for $\cP$.
  The results of \cite{Ramdas2022-jn} assume a trivial initial $\sigma$-algebra, so we prepend the
  trivial $\cF_{-1}$ to the filtration and set $E_{-1}=0$.
  Theorem~14 in \cite{Ramdas2022-jn} then states that whenever $\cP$ is fork-convex and locally dominated,
  every e-process for $\cP$ is dominated by a non-negative process $(L_t)_{t\geq-1}$ with $L_{-1}\leq1$ that
  is a supermartingale for $\cP$. Since the initial distribution is unconstrained, $\cP$ contains, for every $x_0$,
  a distribution with $X_0=x_0$ almost surely, and so $L_0(x_0)\leq L_{-1}\leq1$, by the supermartingale property.
  Let us define $W_t = L_t + 1 - L_0$.
  Then $(W_t)_{t\geq0}$ is a supermartingale for $\cP$ with $W_0=1$, and since $L_0 \leq 1$, it is non-negative.
  Hence $(W_t)_{t\geq0}$ is a test supermartingale for $\model$ that dominates $(E_t)_{t\geq0}=(\monitor_t/\alpha)_{t\geq0}$
  $\Prob$-almost surely for every $\Prob\in\cP$, i.e., on every history consistent with $\model$.

  To conclude, whenever $\monitor_t=1$, then we have $W_t \geq E_t=\monitor_t/\alpha=1/\alpha$, so the monitor $\ind{W_t \geq 1/\alpha} = 1$.
  \qed
\end{proof}

\subsection{Proofs for \cref{sec:constructions}}
\label{app:proofs-constructions}

\thmRiprWellDefined*
\begin{proof}[Proof of \cref{thm:ripr-well-defined}]
  We need to show that for every $w > 0$ the RIPr construction yields a unique fair bet $e$, positive on $\supp(\cU)$.
  Let $\cU$ be a nonempty polytope and $w\in(0,\infty)^{\cX}$, so that the normalised
  weights $q$ satisfy $\supp(q)=\supp(\cU)$.
  The set $\cU$ is compact and $p\mapsto\KL(q\Vert p)$ is lower semicontinuous,
  so the minimum of $\KL(q\Vert p)$ over $p\in\cU$ is attained.
  Since $\cU$ is convex, it contains some $p$ with $\supp(p)=\supp(\cU)$, for which
  $\KL(q\Vert p)<\infty$; hence the minimum is finite.
  As $\KL(q\Vert p)=\infty$ whenever $p(x)=0$ for some $x\in\supp(q)$, every
  $p^*\in\RIPr(q,\cU)$ satisfies $\supp(p^*)=\supp(\cU)$.
  Theorem~1 in \cite{Lardy2024-gy} then implies that all members of $\RIPr(q,\cU)$
  coincide on $\supp(q)=\supp(\cU)$, so $\RIPr(q,\cU)=\{p^*\}$ is a singleton,
  hence $q/p^*$ is well-defined. Theorem~1 in \cite{Lardy2024-gy} also states that
  $q/p^*$ is a fair bet for $\cU$. Finally, since $\supp(q)=\supp(\cU)$, $q/p^*$ is positive on $\supp(\cU)$,
  which completes the proof. \qed
\end{proof}

\thmRiprCompleteness*
\begin{proof}[Proof of \cref{thm:ripr-completeness}]
  Let $V$ be the finite set of vertices of $\cU$. Since
  $\E_p[e]$ is linear in $p$, a function $e$ is a fair bet for $\cU$ if and only if
  $\E_v[e]\leq1$ for every $v\in V$. For every $x\in\supp(\cU)$ there is $v\in V$ with $v(x)>0$,
  so every fair bet satisfies $e(x)\leq1/v(x)$. Hence the set $\mathcal E_e$ of fair bets
  $e'$ with $e'\geq e$ on $\supp(\cU)$, viewed as vectors in $[0,\infty)^{\supp(\cU)}$, is nonempty, closed,
  and bounded, and $e'\mapsto\sum_{x\in\supp(\cU)}e'(x)$ attains its maximum on $\mathcal E_e$
  at some $e'$.

  Every fair bet $e''\geq e'$ belongs to $\mathcal E_e$, so the maximality of the sum gives
  $e''=e'$ on $\supp(\cU)$; that is, $e'$ is an admissible e-variable for $V$
  \cite[Def.~2.11]{Ramdas2025-gq}. Proposition~2.16 in \cite{Ramdas2025-gq},
  applied to the finite set $V$ and the events $\{x\}$, then gives for every $x\in\supp(\cU)$
  some $v^x\in V$ with $v^x(x)>0$ and $\E_{v^x}[e']=1$.
  The average $p=\frac1{|\supp(\cU)|}\sum_{x\in\supp(\cU)}v^x\in\cU$ thus satisfies $\supp(p)=\supp(\cU)$ and
  $\E_p[e']=1$.

  Let $q(x)=e'(x)p(x)$ for $x\in\supp(\cU)$ and $q(x)=0$ otherwise. Then $q$ is a distribution,
  since $\sum_xq(x)=\E_p[e']=1$, $\supp(q)\subseteq\supp(\cU)$, and $e'=q/p$ on $\supp(\cU)$, since $p>0$
  on $\supp(\cU)$. As $q/p$ is a fair bet for $\cU$ and $p\in\cU$, Proposition~7 in
  \cite{Lardy2024-gy} gives $p\in\RIPr(q,\cU)$. \qed
\end{proof}

\lemIntervalRipr*
\begin{proof}[Proof of \cref{lem:interval-ripr}]
The set $\RIPr(q,\cU)$ is a singleton by the proof of \cref{thm:ripr-well-defined}, as $\supp(q)=\supp(\cU)$. The
clipped form is \citet[Thm.~1, eq.~(21)]{DAmico2014-zg} with
$f_x(p)=-q(x)\log p$. \qed
\end{proof}

\subsection{Proofs for \cref{sec:algorithm}}
\label{app:proofs-algorithm}

\begin{lemma}
\label{lem:app-mixture-lipschitz-regret}
Let $T\geq1$, $L\geq0$, and let $g_1,\ldots,g_T\colon\R^d\to(0,\infty)$ be
functions such that $\log g_t$ is $L$-Lipschitz with respect to the $L^2$
norm for every $t$, and define $g(\theta)=\prod_{t=1}^Tg_t(\theta)$.
Further, let $\sigma>0$ and let $\pi=\mathcal N(0,\sigma^2I_d)$ be a Gaussian
prior on $\R^d$. Then for every $\theta^{*}\in\R^d$ we have
\[
  \log\int_{\R^d}g(\theta)\,d\pi(\theta)
  \geq\log g(\theta^{*})
  -\frac{\Vert\theta^{*}\Vert^2}{2\sigma^2}
  -d\log\left(1+\frac{TL\sigma}{\sqrt d}\right)
  -\frac d2.
\]
\end{lemma}
\begin{proof}
  The proof is a standard variational argument. Let $\rho$ be a Gaussian
  distribution with mean $\theta^{*}$ and covariance $\varepsilon^2I_d$ for
  an $\varepsilon>0$ to be specified later. Then, by the Donsker--Varadhan
  variational formula, we have
  \[
    \log\int_{\R^d}\exp\left(\log g(\theta)\right)d\pi(\theta)
    \geq\E_{\rho}[\log g(\theta)]-\KL(\rho\Vert\pi).
  \]
  Since $\log g_t(\theta)$ is $L$-Lipschitz, the function $\log g(\theta)$ is
  $TL$-Lipschitz, and we have
  \[
    \E_{\rho}[\log g(\theta)]\geq\log g(\theta^{*})-TL\,\E_{\rho}[\Vert\theta-\theta^{*}\Vert_2]
    \geq\log g(\theta^{*})-TL\sqrt{\E_{\rho}[\Vert\theta-\theta^{*}\Vert^2_2]}
    =\log g(\theta^{*})-TL\sqrt d\,\varepsilon,
  \]
  where we used Jensen's inequality for the second inequality. The KL
  divergence between two Gaussians is
  \[
    \KL(\rho\Vert\pi)=\frac{\Vert\theta^{*}\Vert^2}{2\sigma^2}
    +\frac d2\left(\frac{\varepsilon^2}{\sigma^2}-1-\log\left(\frac{\varepsilon^2}{\sigma^2}\right)\right).
  \]
  Combining these, we have
  \[
    \log\int_{\R^d}g(\theta)\,d\pi(\theta)\geq\log g(\theta^{*})-TL\sqrt d\,\varepsilon
    -\frac{\Vert\theta^{*}\Vert^2}{2\sigma^2}
    -\frac d2\left(\frac{\varepsilon^2}{\sigma^2}-1-\log\left(\frac{\varepsilon^2}{\sigma^2}\right)\right).
  \]
  Finally, we choose $\varepsilon$ to make the right-hand side small.
  Put $y=TL\sigma/\sqrt d$ and choose $\varepsilon=\sigma/(1+y)$. Then
  \[
    TL\sqrt d\,\varepsilon
    +\frac d2\left(\frac{\varepsilon^2}{\sigma^2}-1
      -\log\frac{\varepsilon^2}{\sigma^2}\right)
    =d\log(1+y)+\frac{dy^2}{2(1+y)^2}
    \leq d\log(1+y)+\frac d2.
  \]
  Substituting this estimate into the preceding bound proves the claim. \qed
\end{proof}

\begin{lemma}
\label{lem:app-ripr-interval-mu-lipschitz}
Let $\cU\subseteq\dist{\cX}$ be an interval uncertainty set with at least two
elements, let $q,q'\in\dist{\cX}$ be two distributions with
$\supp(q)=\supp(q')=\supp(\cU)$, let $r\in\RIPr(q,\cU)$ and
$r'\in\RIPr(q',\cU)$, and let $\mu,\mu'$ be the smallest corresponding
normalizing constants in \eqref{eq:interval-ripr}; that is,
\[
  \mu=\min\left\{m>0:\sum_{x\in\supp(q)}
    \clip\left(\frac{q(x)}{m},\,a_x,\,b_x\right)=1\right\},
\]
and analogously for $\mu'$ and $q'$. Then the following inequality holds:
\[
  \lvert\log\mu-\log\mu'\rvert\leq\max_{x\in\supp(q)}\lvert\log q(x)-\log q'(x)\rvert.
\]
\end{lemma}
\begin{proof}
  For a non-negative vector $p$ supported on $\supp(\cU)$ and $m>0$, define
  \[
    f(p,m)=\sum_{x\in\supp(\cU)}\clip\left(\frac{p(x)}{m},\,a_x,\,b_x\right).
  \]
  Since $\cU$ has at least two elements, $\sum_xa_x<1<\sum_xb_x$.
  Consequently, for every $p$ that is positive on $\supp(\cU)$, the
  continuous non-increasing function $m\mapsto f(p,m)$ has a non-empty
  compact level set $\{m>0:f(p,m)=1\}$, so its smallest element is
  well-defined. Moreover, any two elements of this level set yield the same
  clipped vector. Indeed if $m_1\leq m_2$ then every clipped coordinate at $m_1$
  is at least the corresponding coordinate at $m_2$, and equality of their
  sums forces equality coordinatewise.

  Let $l=\max_{x\in\supp(q)}\lvert\log q(x)-\log q'(x)\rvert$.
  Then for every $x\in\supp(q)$ we have
  \[
    \log q(x)-l\leq\log q'(x)\leq\log q(x)+l,
  \]
  which implies that
  \[
    \frac{q(x)}{e^l}\leq q'(x)\leq q(x)e^l.
  \]
  The function $f(p,m)$ is non-decreasing in $p$ with respect to the
  coordinatewise order, and satisfies $f(cp,cm)=f(p,m)$ for every $c>0$.
  Hence
  \[
    f(q,\mu)=1=f(q',\mu')=f(q'e^l,\mu'e^l)\geq f(q,\mu'e^l).
  \]
  Hence $\mu'e^l$ is either in the same level set of $f(q,\cdot)$ as $\mu$
  or it is strictly larger than $\mu$. By the minimality characterisation above,
  this implies $\mu\leq\mu'e^l$.
  Interchanging $q$ and $q'$ gives similarly $\mu'\leq\mu e^l$, and together
  these imply
  \[
    \lvert\log\mu-\log\mu'\rvert\leq l=\max_{x\in\supp(q)}\lvert\log q(x)-\log q'(x)\rvert.\tag*{\qed}
  \]
\end{proof}

\begin{lemma}
\label{lem:app-ripr-interval-log-payoff-lipschitz}
Let $\cU\subseteq\dist{\cX}$ be a non-empty interval uncertainty set, let
$q,q'\in\dist{\cX}$ satisfy $\supp(q)=\supp(q')=\supp(\cU)$, and let
$r\in\RIPr(q,\cU)$ and $r'\in\RIPr(q',\cU)$. Then, for every
$x\in\supp(\cU)$,
\[
  \left|\log\frac{q(x)}{r(x)}-\log\frac{q'(x)}{r'(x)}\right|
  \leq\max_{u\in\supp(\cU)}|\log q(u)-\log q'(u)|.
\]
\end{lemma}
\begin{proof}
  Write
  \[
    l=\max_{u\in\supp(\cU)}|\log q(u)-\log q'(u)|.
  \]
  If $\cU$ is a singleton, then $r=r'$ and the claim follows directly from
  $|\log q(x)-\log q'(x)|\leq l$.

  Otherwise, let $\mu,\mu'$ be the smallest normalizing constants for
  $q,q'$. By \cref{lem:app-ripr-interval-mu-lipschitz},
  $|\log\mu-\log\mu'|\leq l$. For a fixed $x\in\supp(\cU)$, put
  $z=\log q(x)$ and $z'=\log q'(x)$. The clipped form
  \eqref{eq:interval-ripr} gives
  \[
    \log\frac{q(x)}{r(x)}
    = \clip\left(\log\mu,\,z-\log b_x,\,z-\log a_x\right),
  \]
  where the upper endpoint is understood as $+\infty$ if $a_x=0$; the same
  identity holds for the primed quantities. Each of the three arguments of
  the primed clipping operation differs from its unprimed counterpart by at
  most $l$. Since clipping is monotone in all three arguments and commutes
  with a common translation, we have
  \begin{align*}
    \log\frac{q(x)}{r(x)}+l
    &= \clip\left(\log\mu+l,\,z-\log b_x+l,\,z-\log a_x+l\right)\\
    &\geq \clip\left(\log\mu',\,z'-\log b_x,\,z'-\log a_x\right)
    = \log\frac{q'(x)}{r'(x)}.
  \end{align*}
  Symmetrically, we derive $\log\frac{q'(x)}{r'(x)}+l\geq\log\frac{q(x)}{r(x)}$, which together proves the claim. \qed
\end{proof}

\begin{lemma}
\label{lem:app-ripr-interval-lipschitz}
Let $\cU\subseteq\dist{\cX}$ be a non-empty interval uncertainty set, let
$\feat_x\in\R^d$ with $\|\feat_x\|_\infty\leq B$ for $x\in\cX$, and define
\[
  q_\theta(x)
  =\frac{\ind{x\in\supp(\cU)}\exp(\feat_x\theta)}
    {\sum_{u\in\supp(\cU)}\exp(\feat_u\theta)},
  \qquad
  r_\theta\in\RIPr(q_\theta,\cU),
\]
for $\theta\in\R^d$. Then, for every $x\in\supp(\cU)$, the function
\[
  g_x(\theta)=\log q_\theta(x)-\log r_\theta(x)
\]
is $2B$-Lipschitz with respect to the $L^1$ norm on $\R^d$.
\end{lemma}
\begin{proof}
  For every $u\in\supp(\cU)$, the gradient of the softmax log-probability is
  \[
    \nabla_{\theta}\log q_\theta(u)
    =\feat_u-\sum_{v\in\supp(\cU)}q_\theta(v)\feat_v,
  \]
  and thus
  \[
    \left\|\nabla_\theta\log q_\theta(u)\right\|_\infty
    \leq \max_{v}\|\feat_u - \feat_v\|_\infty\leq 2B.
  \]
  By the standard bounded-gradient characterisation of Lipschitz functions on
  convex domains, for all $\theta,\theta'\in\R^d$,
  \[
    \max_{u\in\supp(\cU)}
      \left|\log q_{\theta'}(u)-\log q_\theta(u)\right|
    \leq 2B\|\theta'-\theta\|_1.
  \]
  Applying \cref{lem:app-ripr-interval-log-payoff-lipschitz} to the
  distributions corresponding to $\theta$ and $\theta'$ now gives
  \[
    |g_x(\theta')-g_x(\theta)|\leq2B\|\theta'-\theta\|_1,
  \]
  as required. \qed
\end{proof}

\begin{lemma}
\label{lem:app-scale-lipschitz}
Let $\cU\subseteq\dist{\cX}$ be a nonempty polytope, let $\feat_x\in\R^d$
with $\|\feat_x\|_\infty\leq B$ for $x\in\cX$, and let
$w^\theta_x=\exp(\feat_x\theta)$. Then, for every $x\in\supp(\cU)$, the
function $\theta\mapsto\log\cC^{\mathrm{Sc}}(w^\theta,\cU)(x)$ is
$2B$-Lipschitz with respect to the $L^1$ norm on $\R^d$.
\end{lemma}
\begin{proof}
We have
\[
  \log\cC^{\mathrm{Sc}}(w^\theta,\cU)(x)
  =\feat_x\theta
   -\log\max_{p\in\cU}\sum_{y\in\supp(\cU)}p(y)\exp(\feat_y\theta).
\]
The term $\feat_y\theta$ is $B$-Lipschitz as a function of $\theta$
for every $y\in\supp(\cU)$; in particular, the first term is $B$-Lipschitz.
For the second term, if $\lVert\theta-\theta'\rVert_1\leq\delta$, we have
$$
w^{\theta'}_y=\exp(\feat_y\theta')\leq\exp(\feat_y\theta+\delta B)\leq e^{\delta B}w^\theta_y.
$$
Hence
$\E_p[w^{\theta'}]\leq e^{\delta B}\E_p[w^\theta]$ for every $p$,
and therefore also $\max_{p\in\cU}\E_p[w^{\theta'}]\leq e^{\delta B}\max_{p\in\cU}\E_p[w^\theta]$.
By the symmetric argument, we also have $\max_{p\in\cU}\E_p[w^{\theta}]\leq e^{\delta B}\max_{p\in\cU}\E_p[w^{\theta'}]$,
and the second term is also $B$-Lipschitz, thus proving the claim. \qed
\end{proof}

\thmMixtureRegret*
\begin{proof}
  Fix a history $h_T= (x_0,a_0,\dots,x_T,a_T)$ and a parameter $\theta^*\in\R^d$.
  If for some $0 \leq t<T$ we have $x_{t+1}\notin\supp(\model(h_t))$, then
  every monitor has infinite value and the claim is trivial.
  So assume that $x_{t+1}\in\supp(\model(h_t))$ for every $0\leq t<T$.

  Let further $g_{t+1}(\theta) = e^\theta(x_{t+1};h_t)$ for $0\leq t < T$ be
  either the RIPr bet or the scale bet.
  By \cref{lem:app-ripr-interval-lipschitz} and
  \cref{lem:app-scale-lipschitz}, $\log g_t(\theta)$ is
  $2B$-Lipschitz with respect to the $L^1$ norm. Since
  $\lVert\theta-\theta'\rVert_1\leq\sqrt d\,\lVert\theta-\theta'\rVert_2$, it
  is therefore $2B\sqrt d$-Lipschitz with respect to the $L^2$ norm. Applying
  \cref{lem:app-mixture-lipschitz-regret} with $L=2B\sqrt d$ gives
  \[
    \log\int_{\R^d}\prod_{t=1}^Tg_t(\theta)\,d\pi(\theta)
    \geq\log\prod_{t=1}^Tg_t(\theta^*)
    -\frac{\lVert\theta^*\rVert^2}{2\sigma^2}
    -d\log(1+2BT\sigma)-\frac d2.
  \]
  The integral on the left is $W_T^{\LM}$, while the product on the
  right is $W_T^{\theta^*}$, proving the claim. \qed
\end{proof}

\corMixtureDetectionTime*
\begin{proof}
  Fix a trajectory on which
  $\liminf_{t\to\infty}\frac1t\log W_t^{\theta^*}\geq g$, and put
  \[
    b_\alpha=\log(1/\alpha),
    \qquad
    R_t
    =\frac{\lVert\theta^*\rVert_2^2}{2\sigma^2}
     +d\log(1+2Bt\sigma)+\frac d2,
    \qquad
    \tau^+_\alpha
    =\inf\{t\geq1:\log W_t^{\theta^*}
                      \geq b_\alpha+R_t\}.
  \]
  By \cref{thm:mixture-regret}, we have for $t=\tau^+_\alpha$,
  \[
    \log W_{\tau^+_\alpha}^{\LM}
    \geq
    \log W_{\tau^+_\alpha}^{\theta^*}
    -
    R_{\tau^+_\alpha}
    \geq b_\alpha+R_{\tau^+_\alpha}
    -
    R_{\tau^+_\alpha}
    =b_\alpha.
  \]
  Hence $\tau_\alpha^{\LM}\leq\tau^+_\alpha$.
  Now fix $\varepsilon\in(0,g/2)$. Since $R_t=o(t)$ and due to the growth
  condition, there exists $t_0$ such that for all
  $t\geq t_0$, $\log W_t^{\theta^*}\geq t(g-\varepsilon)$ and
  $R_t\leq\varepsilon t$, and hence
  $\log W_t^{\theta^*}-R_t\geq t(g-2\varepsilon)$.
  For $\alpha\in(0,1)$, the time
  $t_\alpha=\max\{t_0,\lceil b_\alpha/(g-2\varepsilon)\rceil\}$ satisfies
  $t_\alpha\geq t_0$ and $t_\alpha(g-2\varepsilon)\geq b_\alpha$, so
  $\log W_{t_\alpha}^{\theta^*}-R_{t_\alpha}\geq b_\alpha$.
  As $\tau^+_\alpha$ is the first time with this property,
  $\tau^+_\alpha\leq t_\alpha\leq\max\{t_0,1+b_\alpha/(g-2\varepsilon)\}$. Since $b_\alpha$ tends to infinity as $\alpha\downarrow0$,
  we conclude that
  \[
    \limsup_{\alpha\downarrow0}
    \frac{\tau_\alpha^{\LM}}{T_\alpha}
    =
    g\limsup_{\alpha\downarrow0}
    \frac{\tau_\alpha^{\LM}}{b_\alpha}
    \leq
    g\limsup_{\alpha\downarrow0}
    \frac{\tau^+_\alpha}{b_\alpha}
    \leq
    \frac{g}{g-2\varepsilon}.
  \]
  Since $\varepsilon$ can be chosen arbitrarily small, the left-hand side is at most $1$. \qed
\end{proof}

\end{document}